\documentclass[english,thm-restate, 
]{lipics-v2021}

\nolinenumbers
\hideLIPIcs 

\usepackage{cancel}
\usepackage{ulem}

\usepackage{tikz} 
\usetikzlibrary{patterns}
\usetikzlibrary{shapes.geometric, calc}

\usepackage{hyperref}
\usepackage{color}
\usepackage{tcolorbox}
\definecolor{darkgreen}{RGB}{0,100,0}
\definecolor{firebrick}{RGB}{178,34,34}

\graphicspath{{./pictures/}}

\theoremstyle{plain}

\DeclareMathOperator*{\supp}{supp}

\newcommand{\yy}{y_1}
\newcommand{\yx}{y_2}

\newcommand{\ud}{\mathrm{d}}

\newcommand\R{\mathbb{R}}
\newcommand{\M}{\mathcal{M}}

\newcommand{\Su}{{\mathcal{S}}}

\newcommand{\Tan}{\operatorname{Tan}}
\newcommand{\Nor}{\operatorname{Nor}}
\newcommand{\rch}{\mathrm{rch}}

\newcommand\abs[1]{\left\lvert#1\right\rvert}

\renewcommand{\P}{\mathsf{P}}
\newcommand{\p}{\mathsf{p}}
\newcommand{\q}{\mathsf{q}}

\newcommand{\defunder}[1]{\underset{\text{def.}}{#1} \:}

\title{A free lunch: manifolds of positive reach can be smoothed without decreasing the reach 
}
\titlerunning{Smoothing manifolds of positive reach
}

\author{Hana Dal Poz Kou\v{r}imsk\'a}{University of Potsdam \\{[Potsdam, Germany]}}{hana.dal.poz.kourimska@uni-potsdam.de}{https://orcid.org/0000-0001-7841-0091}{Supported by the DFG project No. 524578210.} 

\author{Andr{\'e} Lieutier}{No affiliation\\{[Aix-en-Provence, France]}}{andre.lieutier@gmail.com }{}{}

\author{
	Mathijs Wintraecken}{Inria Sophia Antipolis, Universit{\'e} C{\^o}te d'Azur\\{[Sophia Antipolis, France]}  }{mathijs.wintraecken@inria.fr}{https://orcid.org/0000-0002-7472-2220}{Supported by the European Union's Horizon 2020 research and innovation programme under the Marie Sk{\l}odowska-Curie grant agreement No. 754411, the Austrian science fund (FWF) grant No. M-3073, the welcome package from IDEX of the Universit{\'e} C{\^o}te d'Azur, and the French National Science Agency (ANR) under the StratMesh grant. }

\authorrunning{Hana Dal Poz Kou\v{r}imsk\'a, Andr{\'e} Lieutier 
	and Mathijs Wintraecken} 

\Copyright{Hana Dal Poz Kou\v{r}imsk\'a,  
	Andr{\'e} Lieutier, and Mathijs Wintraecken} 

\acknowledgements{We thank Jean-Daniel Boissonnat for discussion.  We would also like to acknowledge the organizers of the workshop on `Algorithms for the Medial Axis', and Erin Chambers in particular for giving an impulse to this research. We further thank Victor Bangert for encouragement. }

\ccsdesc{Theory of computation $\rightarrow$ Computational geometry}

\keywords{Reach, Manifolds, Smoothing, Differentiability, Differential topology}

\EventEditors{John Q. Open and Joan R. Access}
\EventNoEds{2}
\EventLongTitle{41st International Symposium on Computational Geometry (SoCG 2025)}
\EventShortTitle{SoCG 2025}
\EventAcronym{SoCG}
\EventYear{2025}
\EventDate{June ???, 2025}
\EventLocation{Kanazawa, Japan}
\EventLogo{socg-logo.pdf}
\SeriesVolume{????}

\begin{document}
	
	\maketitle

	\begin{abstract}
		
		
		Assumptions on the reach are crucial for ensuring the correctness of many geometric and topological algorithms, including triangulation, manifold reconstruction and learning, homotopy reconstruction, and methods for estimating curvature or reach. 
		However, these assumptions are often coupled with the requirement that the manifold be smooth, typically at least $C^2$.
		
		In this paper, we prove that any manifold with positive reach can be approximated arbitrarily well by a $C^\infty$ manifold without significantly reducing the reach, 
by employing techniques from differential topology --- partitions of unity and smoothing using convolution kernels. 
		
		{This result implies that nearly all theorems established for $C^2$ manifolds with a certain reach naturally extend to manifolds with the same reach, even if they are not $C^2$, for free!}
	\end{abstract}

	\section{Introduction}
	
	\subparagraph{What is the reach?}
	The \textit{reach} of a set is a number that captures the geometric properties of its \textit{shape}. Roughly speaking, it provides a bound on the set’s curvature and quantifies how far apart different parts of the set are from each other. As a key descriptor of a shape's complexity, the reach plays a crucial role as an assumption in many geometric and topological algorithms. 
	
	Formally, the reach of a (closed) set $\Su \subset \mathbb{R}^d$ is the minimum of the distance between $\Su$ and its \textit{medial axis}, that is, the set of points in $\mathbb{R}^d$ for which the closest point in $\Su$ is not unique. We illustrate these notions in Figure \ref{fig:reach_med_axis}. 
	
	\begin{figure}[h!]
		\centering
		\includegraphics[width=0.6\textwidth]{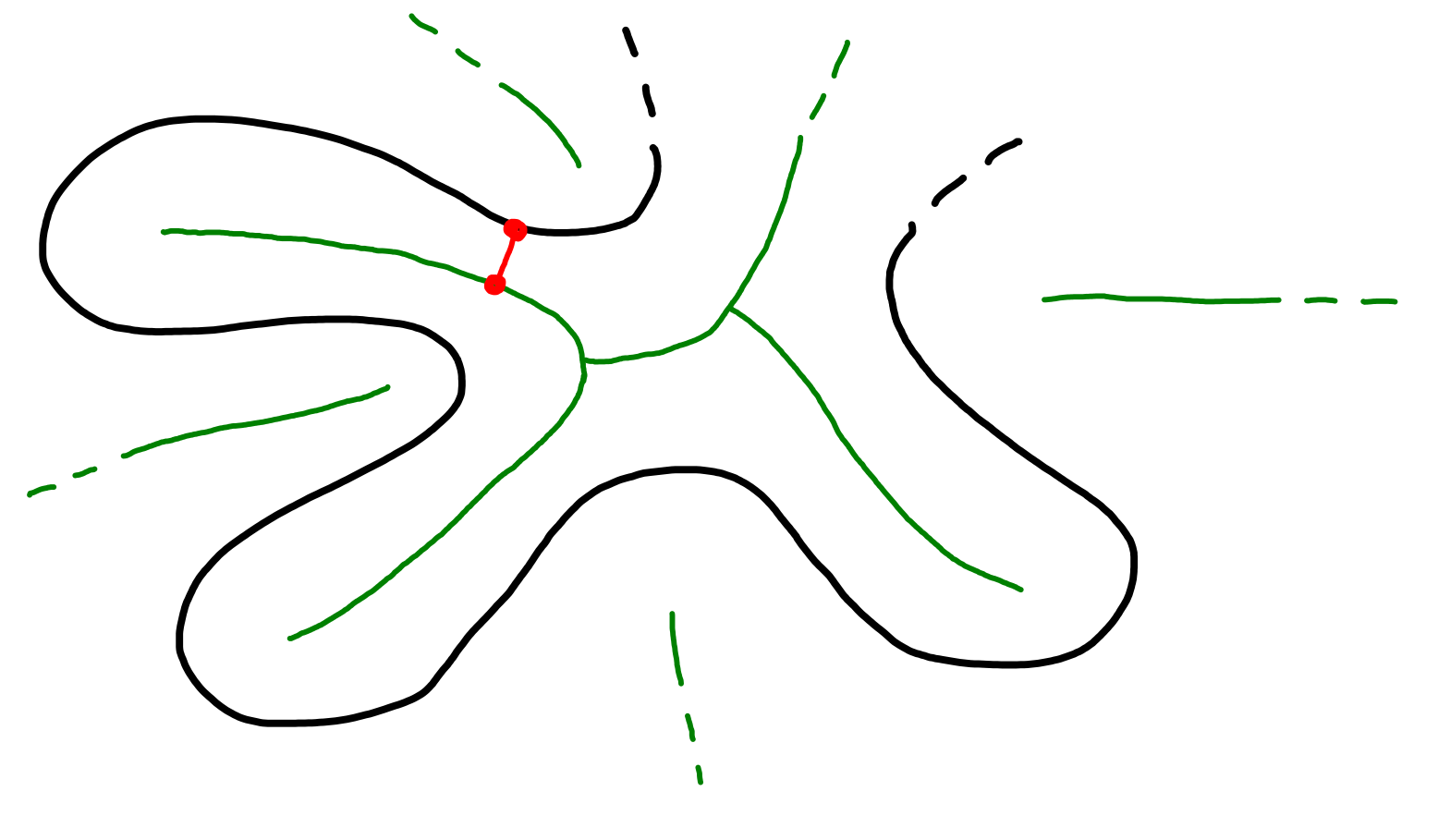}
		\caption{The medial axis (green) of a curve (black) in the plane. The reach is indicated in red.}
		\label{fig:reach_med_axis}
	\end{figure}
	
	\subparagraph{The {early} history of the reach} 
	The reach was first introduced by Federer in \cite{Federer}. Notably, earlier work by Erd\H{o}s explored what we now refer to as the \textit{medial axis}, although it did not address the reach itself \cite{erdos1945some, erdos1946}. While Erd\H{o}s studied the medial axis and Federer considered its complement, the term `medial axis' itself was coined only later, by Blum \cite{blum1967transformation}. A related notion, the \textit{cut locus} in Riemannian geometry, has a significantly longer history, with its origins traced to the work of Poincaré \cite{poincare1905lignes}, Whitehead \cite{whitehead1935covering}, and Myers \cite{myers1935connections, myers1936connections}\footnote{See \cite{HistoryCut} for a nice overview of the early history of the cut locus.}.
	
	Due to their wide applicability, these concepts have been reintroduced multiple times. For instance, the medial axis was reintroduced as the \textit{central set} by Milman and Waksman \cite{milman1980topological}, its complement as the \textit{unique footprint set} by Kleinjohann \cite{kleinjohann1981nachste}, and the reach was referred to as the \textit{condition number} by Niyogi, Smale, and Weinberger~\cite{niyogi2008}.
	
	\subparagraph{The reach and differentiability} 
	In \cite{Federer}, Federer established that the reach is stable under $C^{1,1}$-diffeomorphisms of the ambient space. Here, $C^{1,1}$ denotes a $C^1$ map whose derivative is Lipschitz, and by a $C^{1,1}$-diffeomorphism, we mean that both the diffeomorphism and its inverse are $C^{1,1}$. Federer also mentioned, without extensive detail \cite[Remark 4.20]{Federer}, that the graph of a function has positive reach if and only if the function itself is $C^{1,1}$. Lytchak \cite{lytchak2004geometry, lytchak2005almost} later proved that a topological submanifold of the Euclidean space without boundary has positive reach if and only if it is a $C^{1,1}$-submanifold. A quantified version of this statement can be found in \cite{CompanionPaperC11}.
	
	%
	
	\subparagraph{The reach in geometric and topological algorithms}
	As mentioned, the reach encapsulates the geometric complexity of a shape in a single (non-negative) value, making it a crucial assumption for ensuring the correctness of many geometric and topological algorithms. Several key classes of algorithms that depend on reach assumptions include:
	\begin{itemize} 
		\item triangulation algorithms for surfaces and manifolds, see for example \cite{Amenta1999, amenta1998surface, amenta1998new, JDbook, ChengDeyShewchuck, Dey}, 
		\item manifold learning or reconstruction and manipulation, see for example \cite{Eddie2018stability, belkin2008discrete, fefferman2019GeometricWhitneyProblem, FeffermanFitting, sober2019manifold},
		\item homotopy inference, see for example \cite{niyogi2008, WangWang2020}.
	\end{itemize} 
	In addition, reach estimation in and by itself is an important topic in inference \cite{aamari2023optimal,aamari2019estimating,  berenfeld2022estimating}.  
	
	Most of the aforementioned papers assume that \textbf{the manifold in question is at least $C^2$}, in addition to having positive reach. This assumption is often made because it allows the use of the full machinery of differential and Riemannian geometry. For instance, the second fundamental form is always well-defined in the $C^2$ setting \cite{DoCarmo,docarmo1992}. Yet, this condition is not entirely natural. Manifolds with positive reach are indeed at least $C^{1,1}$, and by Rademacher's theorem \cite{FedererBook}, they are $C^2$ almost everywhere, but not necessarily everywhere.
	
	On the other hand, the $C^{1,1}$ setting is quite natural, as it encompasses many configurations commonly found in modern modeling software, such as computer-aided design (CAD). 
	This is because the majority of manufactured objects can be modeled as $C^{1,1}$ surfaces.
	
For instance, consider a line segment and a circular arc intersecting at a point where their tangents coincide, as illustrated in Figure \ref{fig:CAD}. This configuration has positive reach but is only $C^{1,1}$.   
	\begin{figure}[h!]
		\centering
		\includegraphics[width=0.35\textwidth]{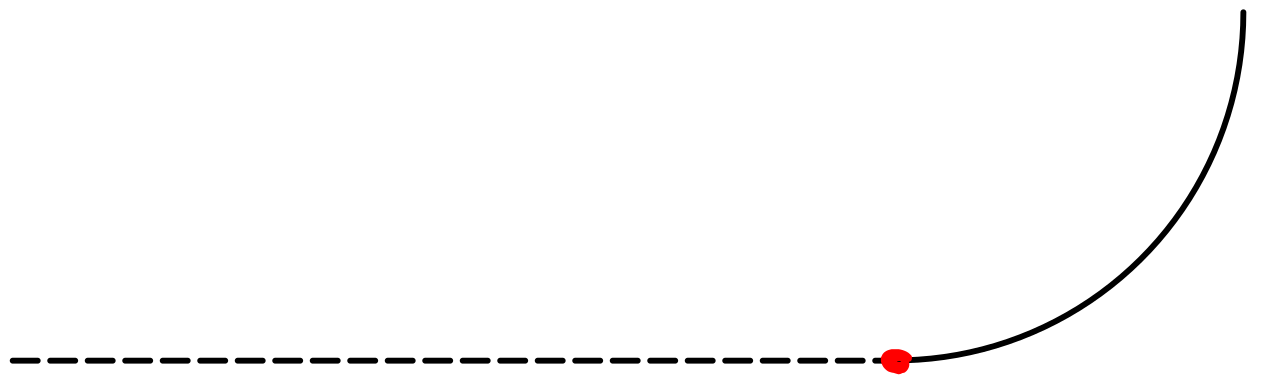}
		\caption{ A $C^{1,1}$ transition (in red) between a circular arc and a straight line segment.}
		\label{fig:CAD}
	\end{figure}

{	\bf The motivation of this paper is to extend ``for free'' all aforementioned  results from $C^2$ manifolds with positive reach to  arbitrary manifolds with positive reach; hence the paper title. 
}	This goal is achieved in Theorem \ref{theorem:MainTheorem}.

The main results from this paper might appear trivial to experts in differential topology at first glance: It is well known that Lipschitz functions can be smoothed without reducing the Lipschitz constant {nor  Lipschitz constants for derivatives} —a fact that is straightforward and proved succinctly in Lemmas~\ref{lem:SmoothingConservesLipschitzConstants} and \ref{lem:smoothingPreservesLipschitzConstantDerivative}. However, this observation alone is far from sufficient to achieve the main result of the paper.
	
	\subparagraph{Our settings} 
	In this work, we identify an embedded manifold locally (in a neighborhood of a point $p$) as a graph of a map from its tangent space to its normal space at $p$. This approach, which we illustrate in Figure \ref{fig:setting_general}, enables us to compare homeomorphic manifolds ($\M$ and $\M'$ in the figure) embedded in Euclidean space of the same dimension by comparing their corresponding maps and derivatives ($f$ and $F$ in the figure). We say that two such manifolds are close \textit{in the $C^1$ sense}, if, roughly speaking, both the manifolds and their tangent spaces are close. Formally, this is a condition on the corresponding maps, and we explain it in Definition~\ref{def:Ck_norm}.
	
	\begin{figure}[h!]
		\centering
		\includegraphics[width=\textwidth]{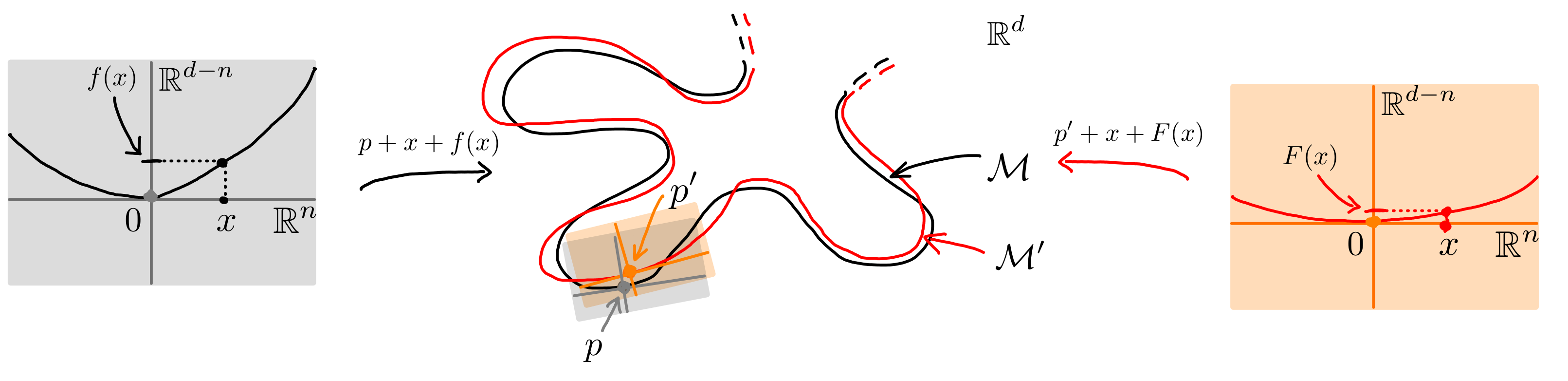}
		\caption{In our setting, we view manifolds locally as graphs of functions. }
		\label{fig:setting_general}
	\end{figure}
	
	\subparagraph{Our contribution} 
	Our main contribution is the following statement:
	\begin{restatable}{theorem}{MainTheorem}
		\label{theorem:MainTheorem}
		Let $\M \subset \mathbb{R}^d$ be a compact manifold of (positive) reach $R$, and $\varepsilon>0$. Then there exists a $C^\infty$ manifold $\M'$ such that:
		\begin{itemize} 
			\item $\M$ and $\M'$ are $\varepsilon$-close as embedded manifolds in the $C^1$ sense.
			\item The reach $R'$ of $\M'$ satisfies $R'\geq R-\varepsilon$. 
		\end{itemize} 
	\end{restatable}
	
	\begin{remark} We can either remove the compactness assumption (because we only need a locally finite cover in the proof of the theorem) or we can assume that the reach is not decreased by $\varepsilon$. 
More precisely, in the compact case, we can increase the reach $R'$ of the manifold $\M'$ by enlarging it (by $\mathcal{O}(\varepsilon)$), and achieve $R'\geq R$. However, this may increase the distance between $\M$ and $\M'$ in the $C^1$ sense (by $\mathcal{O}(\varepsilon)$). 
	\end{remark} 

	Our result can be restated in terms of density in the space of submanifolds:
	\begin{corollary}
		The space of $C^\infty$ embedded submanifolds of $\mathbb{R}^d$ with reach $R$ is dense (in the $C^1$ topology) in the space of $C^{1,1}$ embedded submanifolds of $\mathbb{R}^d$ of reach $R$. 
	\end{corollary}
	
%
%

	\subparagraph{Outline}
	The structure of the paper closely follows the different steps in the proof of Theorem~\ref{theorem:MainTheorem}. Let us provide you with an outlook:
	
	\noindent
	\textbf{\underline{Step 1:}} We start with a compact submanifold~$\M$ of $\R^d$ of positive reach $R$.
	For each point $p\in\M$, we can find a neighbourhood in which $\M$ is a \textit{graph of a function from the (affine) tangent to the normal space at $p$}. 
As it turns out, this function is $C^{1,1}$, and we use the bound on the Lipschitz constant of its derivative to control the angles between nearby tangent spaces. 
We recall the relevant statements in Section~\ref{sec:manifolds_pos_reach}.
	
	We fix an $\varepsilon$ and select a sample of points $p\in\M$ 
	whose neighbourhoods cover $\M$ (see Figure~\ref{fig:step1}). We only work with one neighbourhood at a time.
	\begin{figure}[h!]
		\centering
		\includegraphics[width=.9\textwidth]{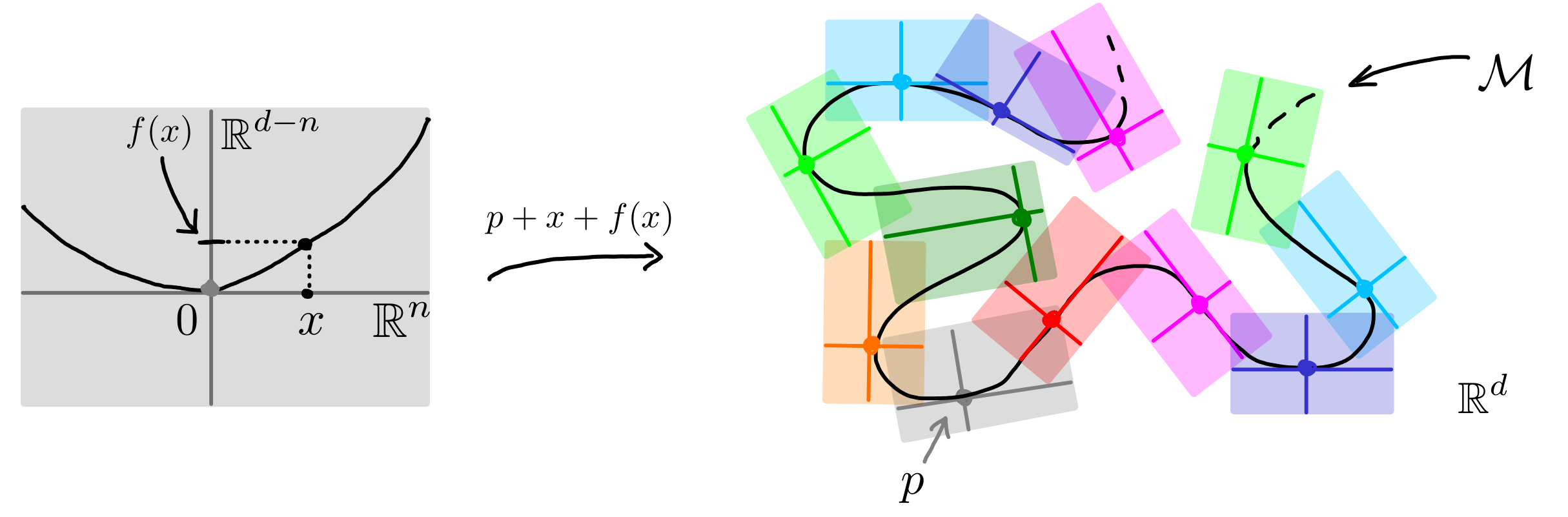}
		\caption{First we cover the manifold $\M$ with neighbourhoods in each of which $\M$ is representable as a graph of a function.}
		\label{fig:step1}
	\end{figure}
	
	\noindent
	\textbf{\underline{Step 2:}} We identify the point $p$ with $0\in\R^d$, and the tangent and normal spaces $T_p\M$ and $N_p\M$ with the first $n$ and last $d-n$ coordinates of $\R^d$, respectively. Following this identification, we denote the map describing $\M$ by $f$, and its domain by $U$.
	
	Our first goal is to smooth $f$ in a neighbourhood of $0$. We split $U$ into three regions: a neighbourhood $U_1$ of $0$, a region $U_3$ covering the viscinity of $\partial U$, and a transition region $U_2$ in between. We then use \textit{kernel-based smoothing} to define a function 
	\[
	F:U\subseteq\R^n\to \R^{d-n}
	\]
	that is smooth in $U_1$ and equals $f$ in $U_3$ (see Figure~\ref{fig:step2}). To achieve this, we employ a \textit{partition of unity function}. We revise the background on smoothing and partitions of unity in Section~\ref{subsec:DifTop}. These techniques allow us to control the Lipschitz constant of $F$ and its derivative, on which we give explicit bounds in Sections~\ref{sec:lipschitz_consts_for_smoothings} and~\ref{sec:interpolation_and_universal_settings}. We use \textit{operator norms} to formulate these bounds.
	The background on operator norms is also presented in Section~\ref{subsec:DifTop}.

	
	\begin{figure}[h!]
		\centering
		\includegraphics[width=.6\textwidth]{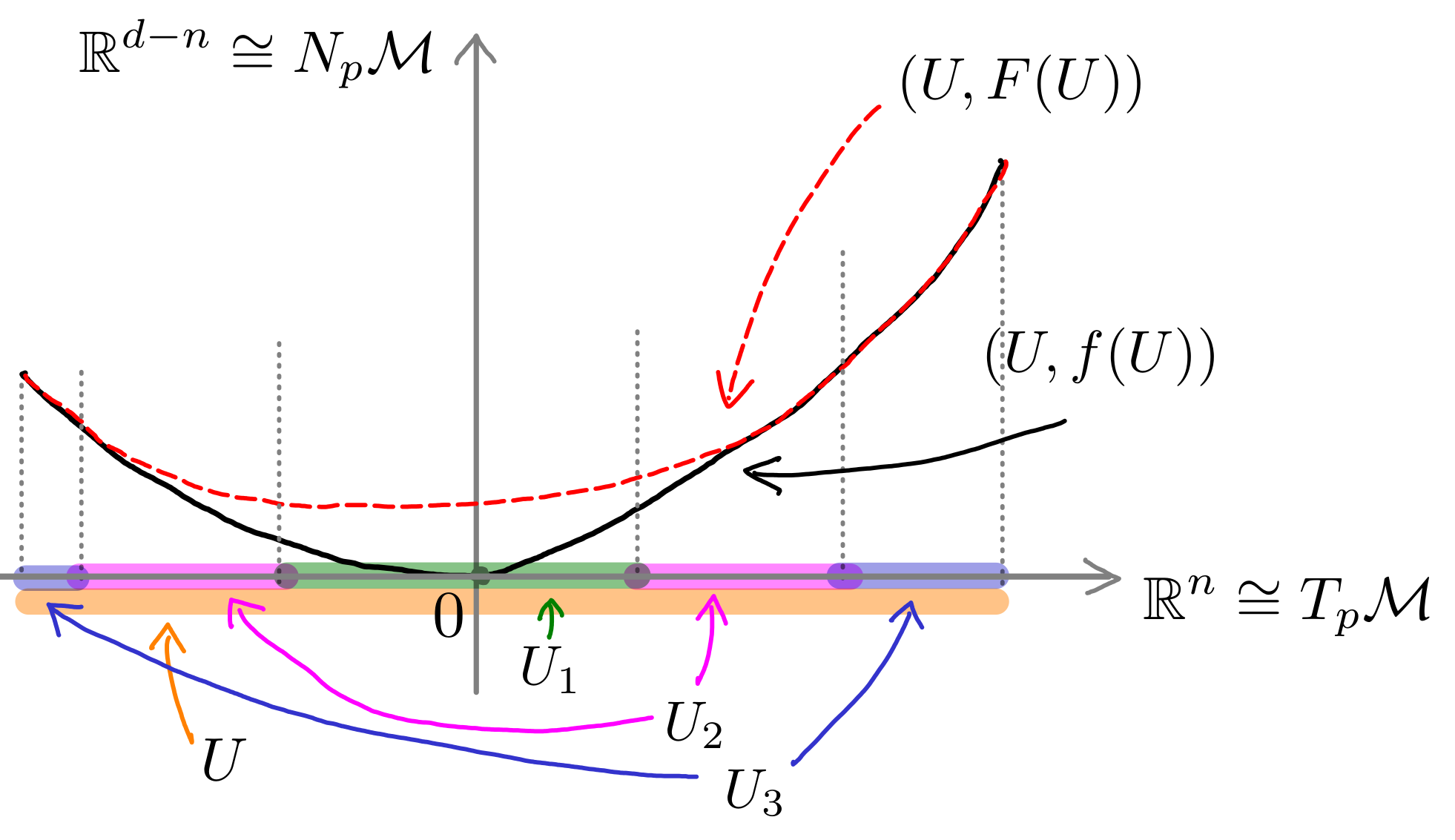}
		\caption{The function $F$ (red) is smooth on the set $U_1$ (green) and equals $f$ on the set $U_3$ (blue).}
		\label{fig:step2}
	\end{figure}
	
	\noindent
	\textbf{\underline{Step 3:}} 
	We perform surgery on the manifold $\M$, and replace the graph of $p+f$ by the graph of $p+F$. We abuse notation and call this manifold $\M'$, although it is `only' smooth in a neighbourhood of the point $p$ for now. Then we estimate the reach of $\M'$. To this end, we leverage a result by Federer (Theorem~\ref{th418Federer}), which characterizes the reach of a manifold through the distance from a point on it to the affine tangent space of another point on it. 
	We pick two points in $p',q'\in\M'$ and investigate the distance from $q'$ to the affine tangent space of $p'$. It turns out that bounding this distance is straightforward when $p'$ does \emph{not} lie in the graph of $p+F$. In the other case, we establish the bound using the relationship between the functions $F$ and $f$, and the Lipschitz constants of $f$ and its derivative. We cover these results in Section~\ref{sec:bounds_on_tangent_spaces}.
	
	\noindent
	\textbf{\underline{Step 4:}} 
	We repeat Steps 2 and 3 iteratively for each point $p$ of our sample, until we have smoothed the whole of $\M$. The process is illustrated in Figure~\ref{fig:step4}. In each iteration, we have a one-parameter freedom in the choice of the smoothing kernel.
	
	In this final step of the proof, we show that both the point sample and the smoothing kernels can be chosen in such a way that at the end, the smooth manifold $\M'$ satisfies the conditions of Theorem~\ref{theorem:MainTheorem}. This final step is described in Section~\ref{sec:proof_main_theorem}.
	
	\begin{figure}[h!]
		\centering
		\includegraphics[width=.9\textwidth]{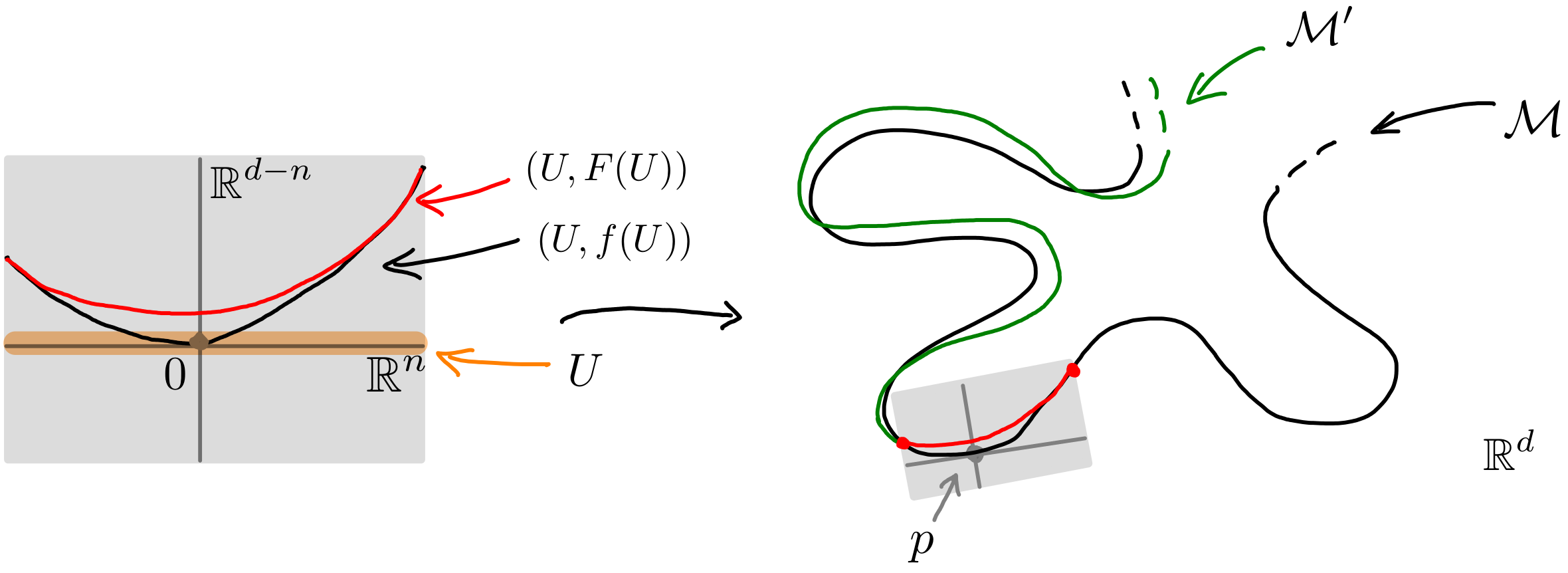}
		\caption{We construct the manifold $\M'$ iteratively. In each neighbourhood (gray), we replace the original manifold ($(U,f(U))$, in black) by a smooth piece ($(U,F(U))$, in red).}
		\label{fig:step4}
	\end{figure}
	
	The proofs of all statements from Sections~\ref{sec:section_3} and~\ref{sec:proof_main_theorem} can be found in Appendix~\ref{sec:proofs}.
	
	%
	%
	%

	
	\section{Preliminaries}\label{sec:Preliminaries} 
	
	\subsection{Manifolds of positive reach}\label{sec:manifolds_pos_reach}
	
	In this section we recall results from~\cite{CompanionPaperC11} on Lipschitz-continuity of maps from a tangent space into a normal space of a point $p$ in a manifold $\M$.
	We also revise a result by Federer~\cite{Federer} on the relationship between the reach of $\M$ and the distance between a point $q\in\M$ and the (affine) tangent space of another point $p\in\M$.
	
	The results presented below apply more broadly to sets of positive reach and are expressed in terms of tangent and normal \textit{cones} rather than spaces. While these two notions generally differ, they coincide for manifolds of positive reach. Indeed, these manifolds are $C^{1,1}$ manifolds, as established in \cite{lytchak2004geometry, lytchak2005almost} or \cite[Theorem 1]{CompanionPaperC11}. We rely extensively on this equivalence and, for simplicity, refer only to tangent and normal spaces in this article.

	Let $T_p\M$ and $N_p\M$ denote the tangent and normal space at a point $p\in\M$, respectively. 
	We write \[\Tan_p\M = p + T_p\M, \qquad \text{and} \qquad\Nor_p\M = p + N_p\M,\] for the translation of $T_p\M$ and $N_p\M$, respectively, by the vector $p$.

	
	
	\subparagraph{Manifolds with positive reach as union of graphs of $C^{1,1}$ functions}
		We first revise the relevant technical statements from \cite{CompanionPaperC11}, which we leverage throughout our paper. We illustrate the setting 
		in Figure~\ref{fig:lemma_16}.
	
	The fact that manifolds with positive reach are $C^{1,1}$  was previously proven by Lytchak \cite{lytchak2004geometry, lytchak2005almost}. In \cite{CompanionPaperC11}, the authors refine it by providing specific quantitative bounds, including optimal bounds on the angles between nearby tangent spaces.
	
%
%
%
		\begin{figure}[h!]
		\centering
		\includegraphics[width=\textwidth]{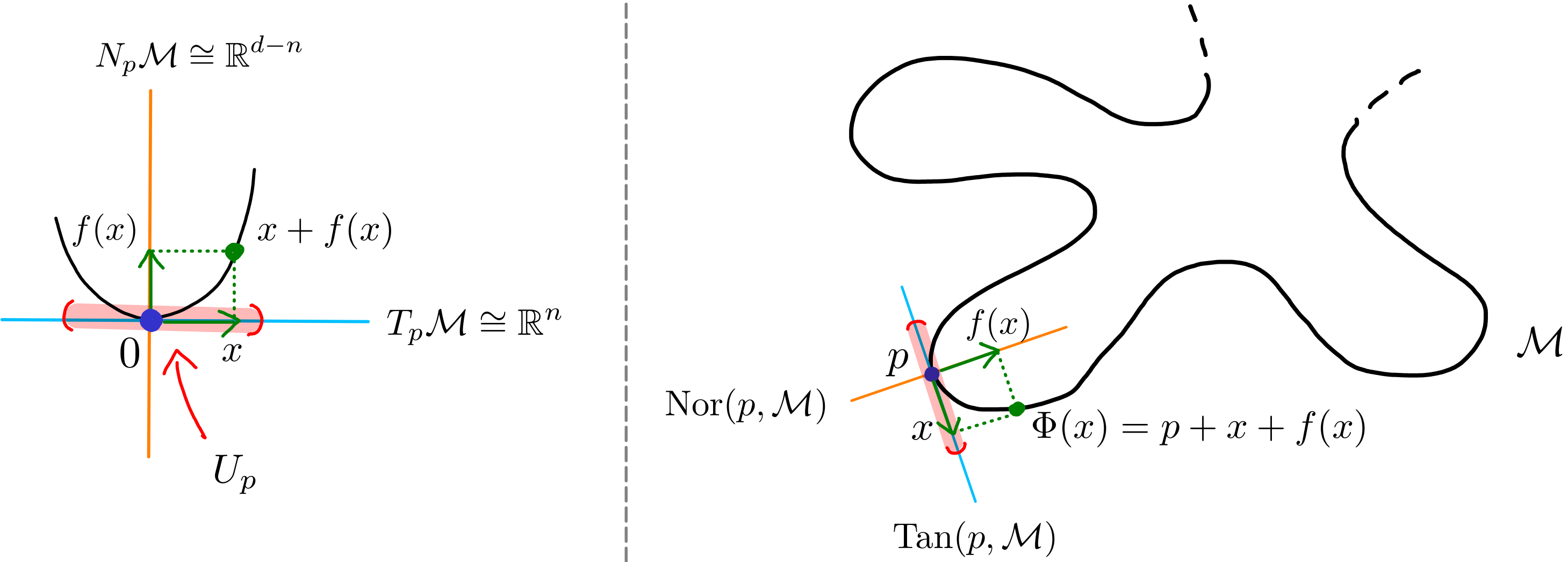}
		\caption{Illustration of the settings of Theorem~\ref{lemma:embeddedManifoldPositiveReachThenC1Embedded}.}
		\label{fig:lemma_16}
	\end{figure}
	
	
	\begin{theorem}[adapted from Theorems 1 and 4 and Lemmas 33 and 37 of {\cite{CompanionPaperC11}}] \label{lemma:embeddedManifoldPositiveReachThenC1Embedded}
		Let $\M$ be a topologically embedded $n$-manifold in $\R^d$ with reach larger than $R>0$, and $p\in \M$ a point in $\M$. Then:
		\begin{enumerate}
			\item There exists an open neighbourhood $U_p\subset T_p\M$ of $p$ containing $B^\circ(p,\rch(\M))$, and a map $f: U_p \rightarrow N_p\M$, such that the map
			\begin{align*}
				\Phi:U_p \to\M, \qquad x\mapsto\Phi(x) \defunder{=} p + x + f(x),
			\end{align*}
			is a $C^{1,1}$ map from $U_p$ to $\M$. 
			\item For every $\delta >0$, there exists an $\alpha >0$ such that on the restricted domain $U_p \cap B(0, \alpha)^\circ$, the derivative of $f$ is $\frac{1}{(R - \delta)}$-Lipschitz. In other words, for all points $y_1,y_2\in U_p \cap B(0, \alpha)^\circ$, the difference between the derivatives at the points $y_1$ and $y_2$ with respect to the operator 2-norm is bounded by
			\[
			\big\|Df(y_2)  - Df(y_1)  \big\|_{2}  \leq \frac{1}{(R - \delta)}  | y_2 - y_1 |.
			\]
			\item If $\delta$ satisfies $ \delta \leq R/2$, then we can choose $\alpha = \sqrt{\delta R}$.
		\end{enumerate}
	\end{theorem}

	\subparagraph{Federer's theorem} Sets of positive reach can be characterized in various ways. In this paper, we focus on the characterization that relates the reach of a manifold to the distance between a point on the manifold and the affine tangent space at another point. The setup is illustrated in Figure~\ref{fig:federer}. 
	
	\begin{theorem}[{adapted from} Theorem 4.18 of \cite{Federer}
]\label{th418Federer} 
		Let $\M\subset\R^d$ be a manifold of positive reach, and $R>0$ a positive number. 
		Then the following two conditions are equivalent:
		\begin{itemize} 
			\item the reach $\rch(\M)$ of $\M$ satisfies $\rch(\M) \geq R$;
			\item every pair of points $p,q \in \M$ satisfies: $ d(q , \Tan_p \M) \leq \frac{|q-p|^2}{2R}$.
		\end{itemize}
	\end{theorem}

\begin{figure}[h!]
	\centering
	\includegraphics[width=0.6\textwidth]{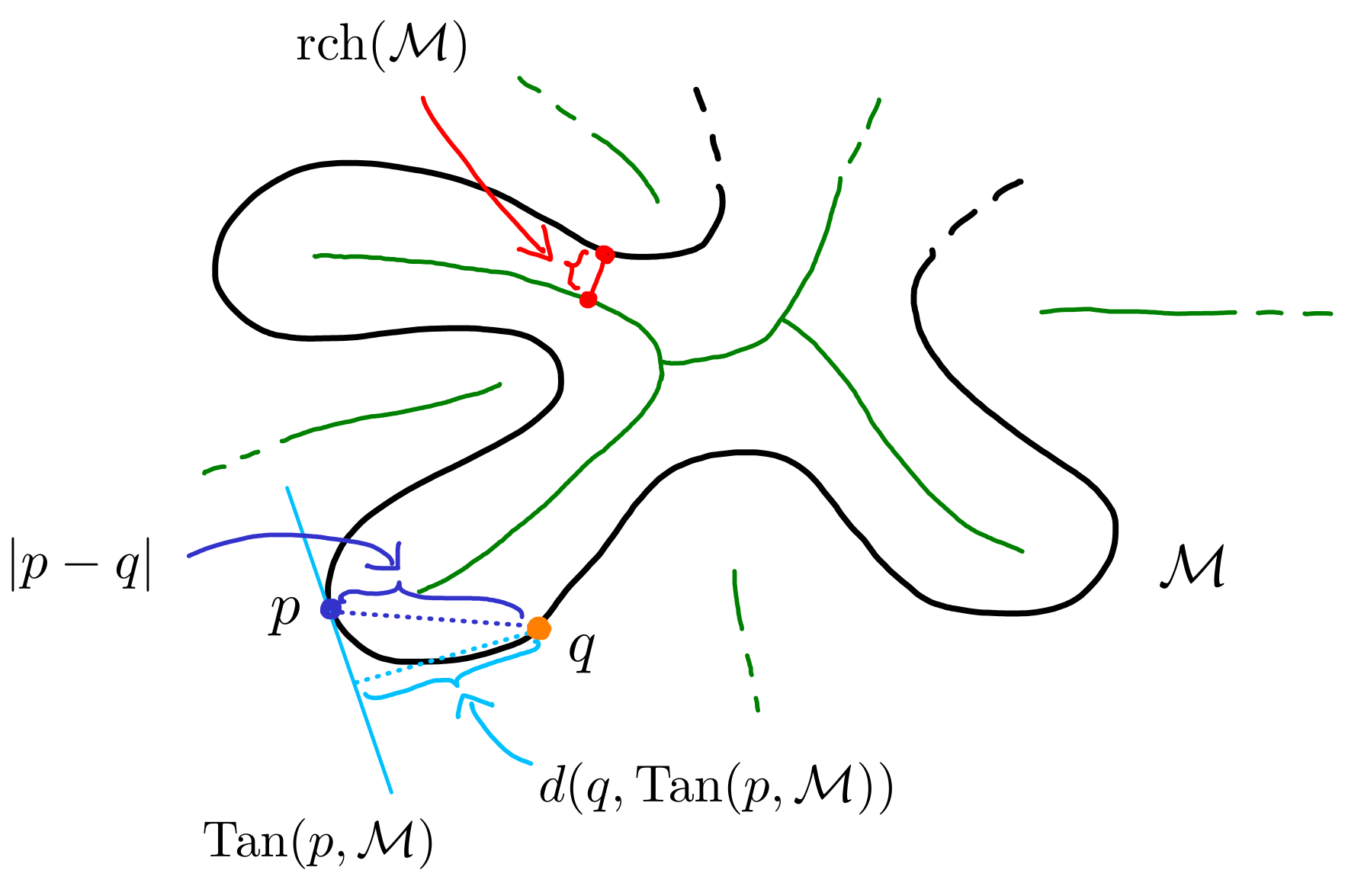}
	\caption{The distance (in light blue) between the point $q\in\M$ and the affine tangent space of the point $p\in\M$ is upper bounded by the squared distance between the points themselves (in dark blue), divided by twice the reach of $\M$.}
	\label{fig:federer}
\end{figure}
	
	\subsection{Results from differential topology}\label{subsec:DifTop}
	In this section we {recall} three {elementary} tools from differential topology: partition of unity functions, the smoothing process, and operator norms. We adopt the formulation and notation used by Hirsch~\cite[Chapter 2]{Hirsch1976}. 
	
	\paragraph*{Partition of unity functions}
	Partition of unity functions allow us to localize constructions and proofs in differential topology. They are defined as follows:
	\begin{definition}\label{def:partition_of_unity}
		Let $M$ be a $C^k$ manifold, with $0 \leq k\leq \infty$, and $\mathcal{U} = \{ U_i\}_ {i \in I}$, with an index set $I$, an open cover of $M$. A $C^k$ \emph{partition of unity} subordinate to $\mathcal{U}$ is a family of $C^k$ maps $\psi_i: M \to [0,1]$, $i \in I$, with the following properties: 
		\begin{itemize} 
			\item For every $i \in I$, the support\footnote{The support of a function $\psi$ is the closure of the set $\psi^{-1} (\mathbb{R}\setminus \{0\})$.} $\supp(\psi_i)$ of $\psi_i$ is contained in the set $U_i$. 
			\item The collection $\{\supp(\psi_i)\}_{i\in I}$ of the supports of $\psi_i$ is locally finite. 
			\item The maps $\psi_i$ sum to the function that is identically equal to $1$, that is,
			\[ \sum_i \psi_i (x) =1 .\] 
		\end{itemize} 
	\end{definition} 
	\begin{figure}[h!]
		\centering
		\includegraphics[width=.5\textwidth]{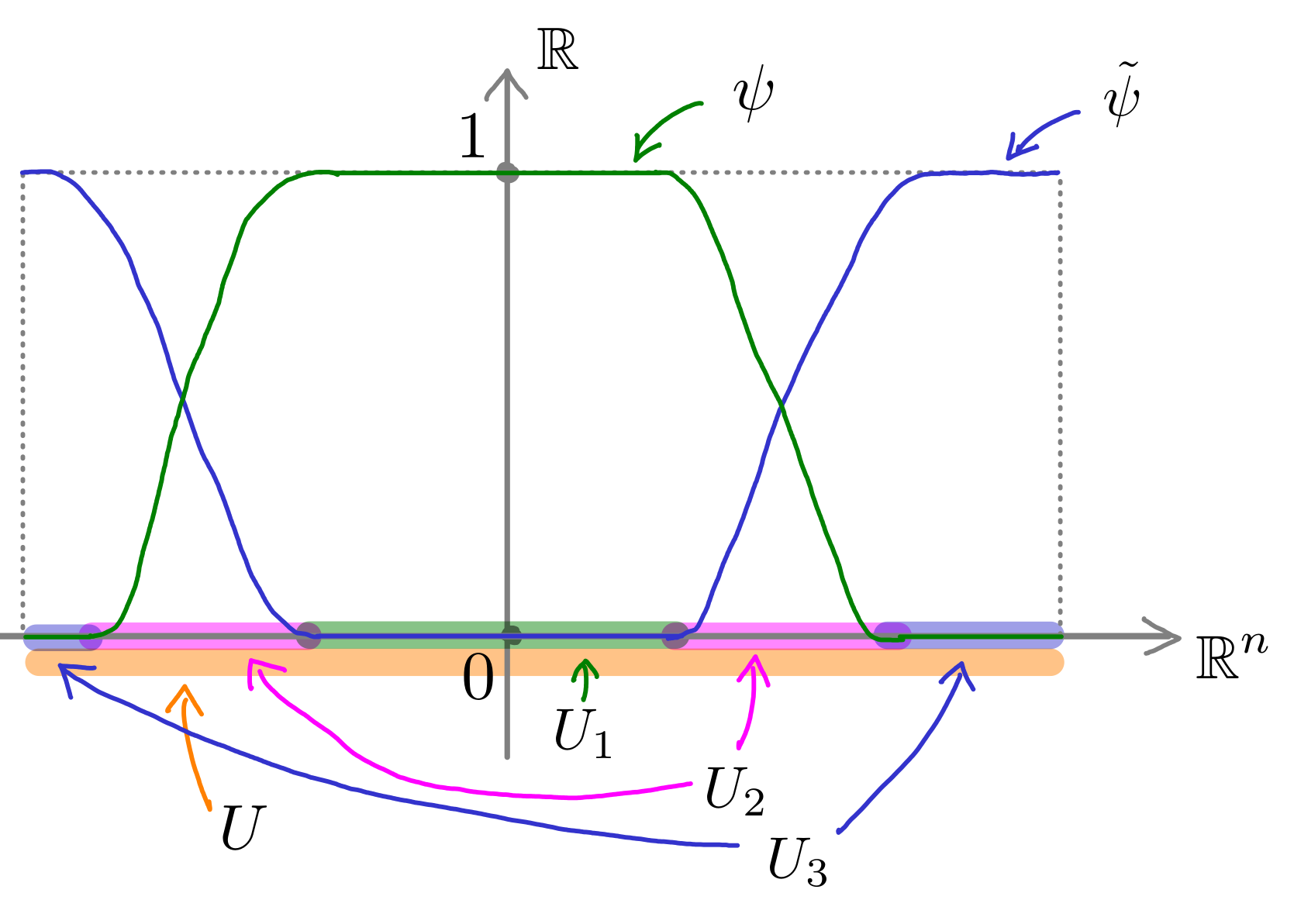}
		\caption{A cover $\{U_1\cup U_2, U_2\cup U_3\}$ of a set $U\subseteq \R^n$, and the corresponding family $\left\{\psi,\tilde{\psi}\right\}$ of partition of unity functions.}
		\label{fig:part_unity}
	\end{figure}
	We illustrate the concept for our setting in Figure~\ref{fig:part_unity}. We can find a partition of unity for any open cover:
	\begin{theorem}[Theorem 2.1 of {\cite{Hirsch1976}}] Let $M$ be a $C^k$ manifold with $1\leq k \leq \infty$. Every open cover of $M$ has a subordinate $C^k$ partition of unity. 
	\end{theorem}  
	
	\paragraph*{Kernel smoothing} 
	Kernel smoothing is, broadly speaking, a process that averages a map using a \textit{kernel} function, making the result at least as smooth as the kernel itself. The kernel is defined as follows:
	\begin{definition}[Smoothing kernel]\label{def:kernel}
		A (smooth) map $\varphi : \mathbb{R}^n \to \mathbb{R}$ is called a \emph{convolution} or a \emph{smoothing kernel} if it is non-negative, has compact support, and $\int_{\mathbb{R}^n} \varphi =1$.
		
		The \emph{support radius} of the smoothing kernel $\varphi$ is the smallest value $\sigma\in\mathbb{R}_{\geq 0}$, for which the support $\supp(\varphi)$ of $\varphi$ is contained in the closed ball of radius $\sigma$ centred at the origin: \[\supp(\varphi)\subset B(0,\sigma)\subset \mathbb{R}^n.\]
	\end{definition}
	We illustrate the smoothing radius in Figure~\ref{fig:support_radius}.
	\begin{figure}[h!]
		\centering
		\includegraphics[width=\textwidth]{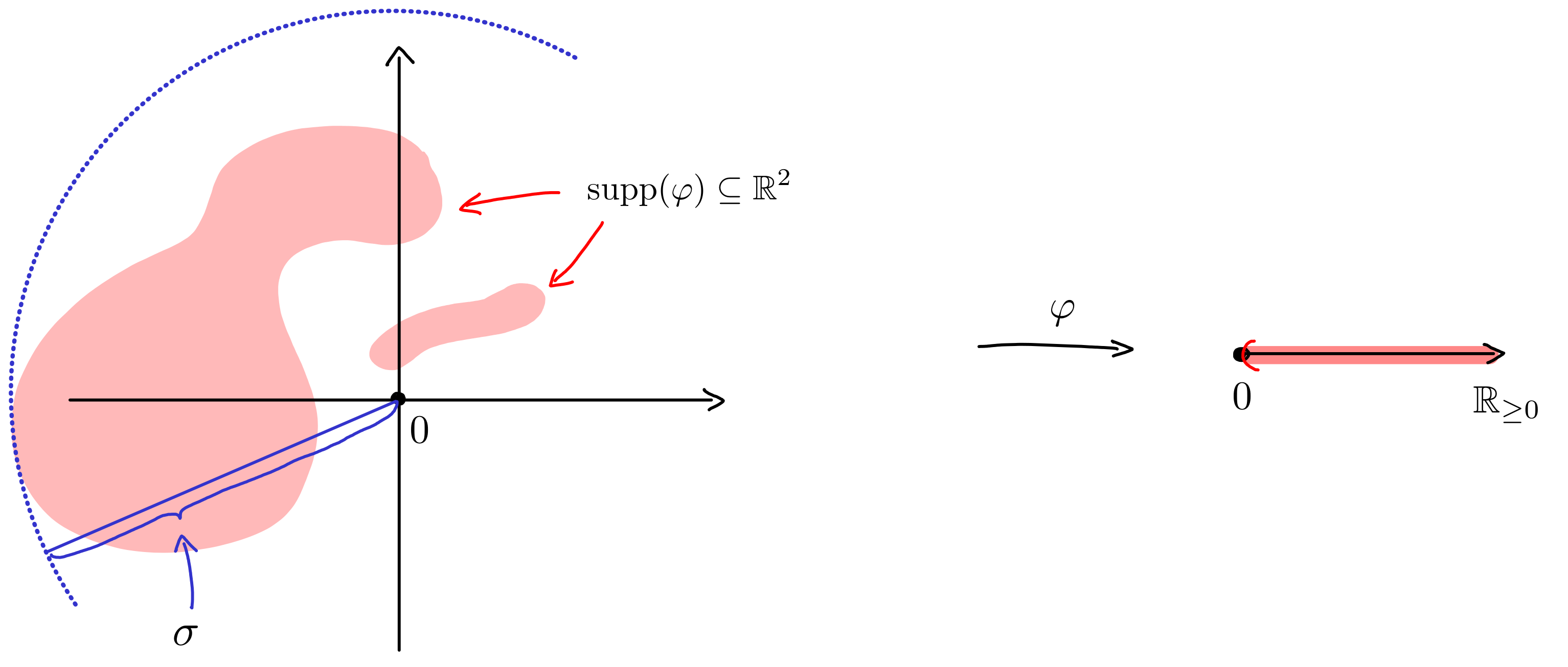}
		\caption{The support radius of the smoothing kernel $\varphi$. 
		}
		\label{fig:support_radius}
	\end{figure}
	
	Smoothing relies on neighbourhoods, determined by the support radius of the kernel. For smoothing of a map on a given set to be well-defined, the map itself must be well-defined on a sufficiently thick neighborhood surrounding the set. Consequently, it is sometimes necessary to shrink the domain of the map where the smoothing will be applied: 
	\begin{definition}\label{def:shrinking} 
		Consider a smoothing kernel with support radius $\sigma$. Given an open set $U\subset \mathbb{R}^n$, the \emph{shrinking} $U_\sigma$ of $U$ is defined as
		\[U_\sigma = \{ x \in U| B(x, \sigma) \subset U \}. \] 
	\end{definition} 
	\begin{figure}[h!]
		\centering
		\includegraphics[width=0.45\textwidth]{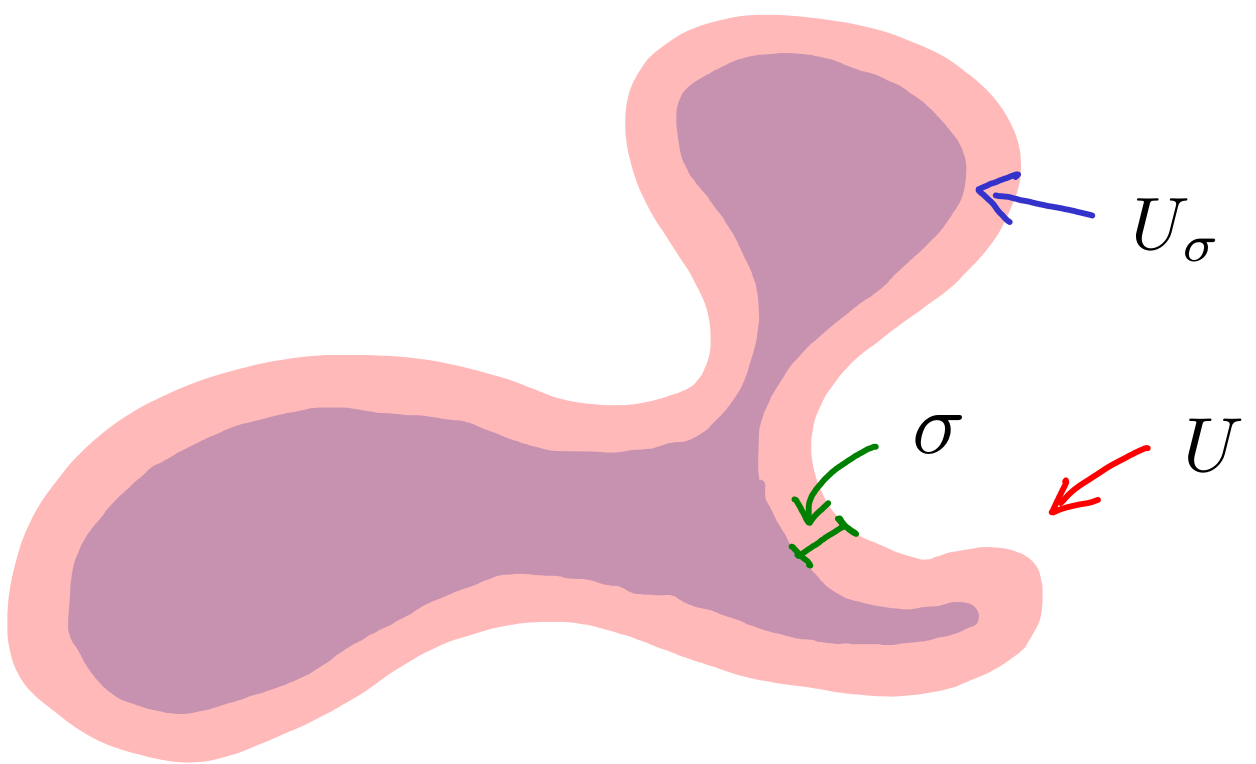}
		\caption{The shrinking of a set as defined in Definition \ref{def:shrinking}.}
		\label{fig:shrinking}
	\end{figure}
	The smoothing process, or in other words, the \textit{convolution}, is carried out through integration:
	\begin{definition}[Convolution]\label{def:convolution}
		Let $\varphi : \mathbb{R}^n \to \mathbb{R}$ be a smoothing kernel with support radius $\sigma$, $U \subset \mathbb{R}^n$ an open set, and $f: U \to \mathbb{R}^{d-n}$ a continuous map. The \emph{convolution} of $f$ by $\varphi$ is the map 
		\begin{align}
			&\varphi \ast f : U_\sigma \to \mathbb{R}^{d-n},
			\nonumber 
			& x\mapsto\varphi \ast f(x) =\int_{\mathbb{R}^n} \varphi (y) f(x-y) \ud y.
		\end{align} 
	\end{definition}
	Smoothing improves the smoothness of the map, and commutes with differentiation: 
	\begin{theorem}[{\cite[Theorem 2.3 (a) and (b)]{Hirsch1976}}]\label{Th23abHirsch}
		Let $\varphi: \mathbb{R}^n \to \mathbb{R}$ be a smoothing kernel with support radius $\sigma>0$, $U \subset \mathbb{R}^n$ an open set and $f: U\to \mathbb{R}^{d-n}$ a continuous map. The convolution $\varphi  \ast f: U_\sigma  \to \mathbb{R}^{d-n}$ has the following properties: 
		\begin{itemize} 
			\item If $\varphi$ is $C^k$, with $1 \leq k\leq \infty$, then so is $\varphi \ast f$, and for each finite $\ell\leq k$,
			\[ D^\ell (\varphi \ast  f ) = D^\ell (\varphi) \ast  f \] 
			on $U_\sigma$.  
			\item If $f$ is $C^{k}$, with $1 \leq k\leq \infty$, then so is $\varphi \ast f$, and for each finite $\ell\leq k$,
			\[ 
			D^\ell (\varphi \ast f)= \varphi \ast (D^\ell f).
			\] 
		\end{itemize} 
	\end{theorem} 
	\paragraph*{Operator norms}
	The last result we need is the convergence of smoothed maps to the original map. We consider convergence of the map itself as well as its first $k$ derivatives with respect to the so-called \textit{$C^k$ norm}. We define this norm in two steps:
	\begin{definition}\label{def:norm_linear_maps}
		The norm $\| .\|$ of a $k$-linear map 
		\[
		S: \mathbb{R}^n \times  \dots \times \mathbb{R}^n \to \mathbb{R}^m, \qquad (u_1,\dots,u_k)\mapsto S(u_1 , \dots , u_k),
		\] 
		is defined as 
		\[ \| S \| = \max_{|u_i |=1 }   | S (u_1 , \dots , u_k) |  .\] 
	\end{definition} 
	Consider a $C^k$ map $f: U  \to \mathbb{R}^{d-n}$, and its $r$th order derivative $D^r  f$.
	For $r \geq 1$, the map $D^r  f (x)$ is an $r$-linear map at each point $x\in U$, and we use the above definition to measure the norm $\| D^r  f (x) \|$ of $D^r  f$ at $x$. For $r=0$, the notation $\| D^r  f (x) \| = \| D^0  f (x) \| $ should be understood as $|f(x)|$. 
	
	\begin{definition}[$C^k$ norm]\label{def:Ck_norm}
		Let $f :U \to \mathbb{R}^{d-n}$ be a $C^k$ map, $U \subset \mathbb{R}^n$ an open set, and $C \subset U$ any subset of $U$. The $C^k$ norm of $f$ {on $C$} is defined as 
		\[ \|f \| _{k,C} = \sup \{ \| D^r  f (x) \| \mid x\in C ,  0 \leq r \leq k\}. 
		\]
	\end{definition} 
	We can now define the convergence we will use: 
	\begin{theorem}[{\cite[Theorem 2.3 (c)]{Hirsch1976}}]\label{Th:Hirsch23c}
		Let $U \subset \mathbb{R}^n$ be an open set with a compact subset $C \subset U$, and $f: U\to \mathbb{R}^{d-n}$ a $C^k$ map, with $ 0 \leq k \leq \infty$.
		For any $\varepsilon> 0$ there exists a value $\sigma >0$ such that $C \subset U_\sigma$, and any $C^k$ smoothing kernel $\varphi$ with support radius $\sigma $ satisfies
		\[ \| \varphi \ast f - f \| _{k, C} \leq \varepsilon. \]   
	\end{theorem} 
	\begin{remark} \label{remark:operator_2-norm}
		The operator $2$-norm of a matrix $A$ is defined as 
		\begin{align} 
			\|A \|_2 =\sup_{v \neq 0} \frac{|A v | }{|v|}= \max_{|v|=1} |A v |  ,  
			\nonumber
		\end{align} 
		where $| \cdot |$ denotes the usual Euclidean ($2$-)norm. 

		For the first derivative of a function $f$, evaluated at a point $x$, the norm from Definition~\ref{def:norm_linear_maps} coincides with the operator $2$-norm, that is, $\|D^1  f (x) \| = \|D^1  f (x) \|_2$.
		In particular,
		\[\|Df \| _{0,C} = { \sup }\{ \|Df (x) \|_2 \mid x\in C \} \qquad \text{and}\qquad \|f \| _{1,C} = {\max} \{\|f \| _{0,C}, \|Df \| _{0,C} \}. \]
	\end{remark} 
	
	In this paper, we primarily focus on the case where $k=1$, as estimating the first derivative proves to be the main challenge in our proofs. Consequently, much of our work involves dealing with operator 2-norms.
	
	In addition, we often make a choice of the map $f$ on which we perform the smoothing. The lemma below implies that our choice does not depend on the map $f$ itself, but only on its Lipschitz constant.
	\begin{lemma}[Folklore] \label{lemma:norm_of_smoothing_depends_on_Lipschitz_const_only}
		We adapt the settings from Theorem~\ref{Th:Hirsch23c}. If the map $f$ is $L$-Lipschitz, then for any point $x\in U_\sigma$ and $y \in B(0, \sigma)$ it holds that $ | f(x-y) -f(x)| \leq L\sigma $.
		As a consequence,
			\begin{align}   
					\| \varphi \ast f - f \| _{0, C}\leq  L \sigma.\nonumber
				\end{align}  
	\end{lemma}

	\section{On the control over Lipschitz constants while smoothing and using partition of unity functions}\label{sec:section_3}
	This section consists of three parts, in which we
	\begin{itemize}
		\item {recall} that smoothing by convolution does not affect Lipschitz constants;
		\item set the stage to define the function $F$ that locally describes our smoothed manifold $\M'$, and determine the Lipschitz constant of $F$ and its derivative;
		\item establish two bounds on the distance between the affine tangent space of a point on the graph of $F$, and another point in $\M'$.
	\end{itemize}
	
	\subsection{Lipschitz constants for smoothings by convolution}\label{sec:lipschitz_consts_for_smoothings}
	In this section we focus on Lipschitz constants. 
	To this end, we denote the Lipschitz constant of a function $g$ by $L_g$.
	We recall that a function $g: \mathbb{R}^n \to \mathbb{R}^m$ is Lipschitz with constant $L_g$ if for all points $y_1,y_2\in \R^n,$
	\begin{align}\label{eq:Lipschitz_cont}
		\abs{g(y_2) - g(y_1)}\leq L_g \abs{y_2-y_1}.
	\end{align}
	Similarly, if $g$ is differentiable, the derivative $Dg$ of $g$ is Lipschitz with constant $L_{Dg}$ if for all points $y_1,y_2\in \R^n,$
	\begin{align}\label{eq:Lipschitz_cont_derivative}
		\big\|Dg(y_2)  - Dg(y_1)  \big\|_{2}  \leq L_{Dg} | y_2 - y_1 |.
	\end{align}
	We first {recall}  that smoothing does not influence the Lipschitz constant of a function or its derivative. 

{The result of the convolution of a function $g$ by a kernel is just a barycenter of translates of $g$. 
	Since the minimal Lipschitz constant satisfied by a function defines a semi-norm on functions, it is a convex functional. As a result:} 
	\begin{lemma}[Folklore] \label{lem:SmoothingConservesLipschitzConstants} 
		Let $g: \mathbb{R}^n \to \mathbb{R}^{d-n}$ be an $L_g$-Lipschitz map and $\varphi: \mathbb{R}^n \to \mathbb{R}$ a smoothing kernel. Then the smoothing by convolution $\varphi \ast g$ is also $L_g$-Lipschitz.
	\end{lemma} 
	{Moreover, since convolution commutes with derivation, we get:}
	\begin{lemma}[Folklore] \label{lem:smoothingPreservesLipschitzConstantDerivative}
		Let $g:\mathbb{R}^n \to \mathbb{R}^{d-n}$ be a function whose \emph{first derivative} is $L_{Dg}$-Lipschitz. 
		Let further $\varphi: \mathbb{R}^n \to \mathbb{R}$ be a smoothing kernel. 
		Then the first derivative of the smoothing by convolution $\varphi \ast g$ is also $L_{Dg}$-Lipschitz.
	\end{lemma}
	
	\subsection{Interpolation between a Lipschitz function and its smoothing by convolution}\label{sec:interpolation_and_universal_settings}
	In this section, we consider an interpolation $F$ between a map $f$ locally describing our manifold, and its convolution $\varphi\ast f$, using a partition of unity function $\psi$. A formal definition of the function $F$ will follow shortly. We illustrate this construction in Figure~\ref{fig:F_construction}.
	
	\begin{figure}[h!]
		\centering
		\includegraphics[width=0.6\textwidth]{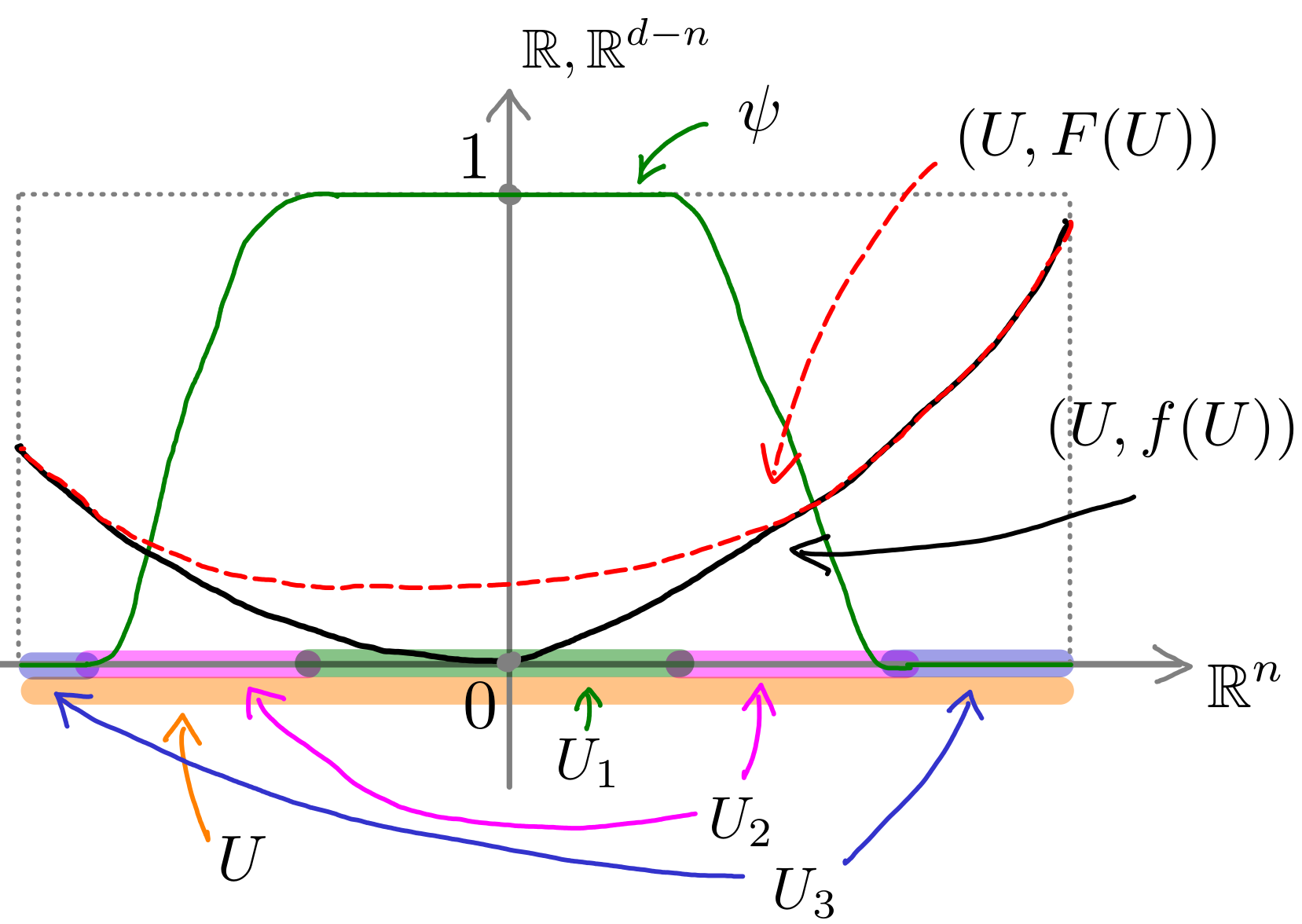}
		\caption{The superposed graphs of $f, F$, and $\psi$ show that $F$ is smooth on $U_1$ and $F=f$ on $U_3$.}
		\label{fig:F_construction}
	\end{figure}
	
	To be able to use local arguments, we need to prove that the Lipschitz constants of $F$ and its first derivative are close to the Lipschitz constants of $f$ and its first derivative.
	
The Lipschitz constants of the function $F$ were studied in \cite{SmoothingOnManifolds2007}, leading to a result similar to our Lemma \ref{Lem:LipschitzWithoutDer}. Our key contribution lies in determining the Lipschitz constant of the \textit{derivative} of $F$, a significantly more intricate task and one that is crucial for proving the main result.
	
	\noindent\textbf{Throughout the rest of this paper, we operate under the following settings:}
	\begin{itemize}
		\item We consider a function $f:U\subseteq\mathbb{R}^n \to \mathbb{R}^{d-n}$ whose graph is (locally) the manifold $\M$. Such a map exists due to Theorem~\ref{lemma:embeddedManifoldPositiveReachThenC1Embedded}. Moreover, both $f$ and its derivative $Df$ are Lipschitz, and we denote their Lipschitz constants by $L_f$ and $L_{Df}$, respectively.
		\item We cover the domain $U$ with two sets $\{U_1\cup U_2, U_2\cup U_3\}$ and consider a partition of unity subordinate to this cover. We denote the partition of unity function corresponding to the set $U_1\cup U_2$ by $\psi$, and the closure of its support by
		\[\overline{\supp{\psi}} = C.\]
		We assume $C$ to be compact. We denote the Lipschitz constants of $\psi$ and its derivative $D\psi$ by $L_{\psi}$ and $L_{D\psi}$, respectively, write
		\begin{align}
			L_{\psi,D \psi}= \max \{L_{\psi} , L_{D\psi} \},
			\label{eq:LipschitzpsiDpsi} 
		\end{align}
		and note that
		\begin{align} 
			\sup_{y\in C} \| D \psi  (y) \|_2 =  \| D \psi \|_{0,C} \leq L_{\psi} \leq L_{\psi,D \psi}.
			\label{eq:boundDpsi} 
		\end{align} 
		\item For $\varepsilon>0$ we choose a smoothing kernel $\varphi_\varepsilon:U\to \R$  such that
		\begin{align} 
			\| \varphi_{\varepsilon}  \ast f - f \| _{1, C}  \leq \varepsilon.
			\label{eq:AssumptionSmoothing1} 
		\end{align} 
		Such kernel exists due to Theorem \ref{Th:Hirsch23c}. 
		We recall (see Remark~\ref{remark:operator_2-norm}) that inequality~\eqref{eq:AssumptionSmoothing1} holds if and only if both of the following inequalities hold:
		\begin{align} 
			\sup_{y\in C}\| (\varphi_{\varepsilon}  \ast Df - Df)(y) \|_2=\| \varphi_{\varepsilon}  \ast Df - Df \| _{0, C}  \leq \varepsilon,
			\label{eq:AssumptionSmoothing2} 
		\end{align} 
		and 
		\begin{align} 
			\sup_{y\in C}\abs{(\varphi_{\varepsilon}  \ast f - f)(y) }=\| \varphi_{\varepsilon}  \ast f - f \| _{0, C}  \leq \varepsilon.
			\label{eq:AssumptionSmoothing3} 
		\end{align}
		\item The (locally) smoothed manifold $\M'$ is then locally described by the graph of the function
		\begin{align}\label{eq:smoothed_mfld_function}
			F:U\to\R^{d-n}, \qquad x\mapsto F(x)= (1- \psi(x))  f (x) + \psi (x)(\varphi_\varepsilon \ast f(x)).
		\end{align}
	\end{itemize}

	\begin{lemma}\label{Lem:LipschitzWithoutDer} 
		The function $F$ is Lipschitz on $C$ with Lipschitz constant $(L_\psi  \varepsilon +  L_f )$.
	\end{lemma}
	
	\begin{lemma} \label{lem:SmoothingWithPartition} 
		The first derivative $DF$ of the function $F$ is Lipschitz on $C$ with Lipschitz constant $(3 L_{\psi,D \psi} \varepsilon   +L_{Df})$.
	\end{lemma} 
	
	\subsection{On bounding the reach of $\M'$ in terms of Lipschitz constants}\label{sec:bounds_on_tangent_spaces}
	Let $\M'$ denote the manifold that 
	\begin{itemize}
		\item equals the graph of the function $p+F$ inside the neighbourhood of the point $p$;
		\item equals $\M$ outside of this neighbourhood. 
	\end{itemize}
	
	Our next goal is to bound the reach of $\M'$. To this end, we use a result by Federer~\cite{Federer} that characterizes the reach of a manifold through the distance from a point on it to a tangent space of another point on it. We recalled this result in Theorem~\ref{th418Federer}.
	We provide two different bounds on this distance. In order to prove the latter, we also establish bounds on the angle between tangent spaces of the graphs of $f$ and $F$.
	
	\begin{figure}[h!]
		\centering
		\includegraphics[width=.75\textwidth]{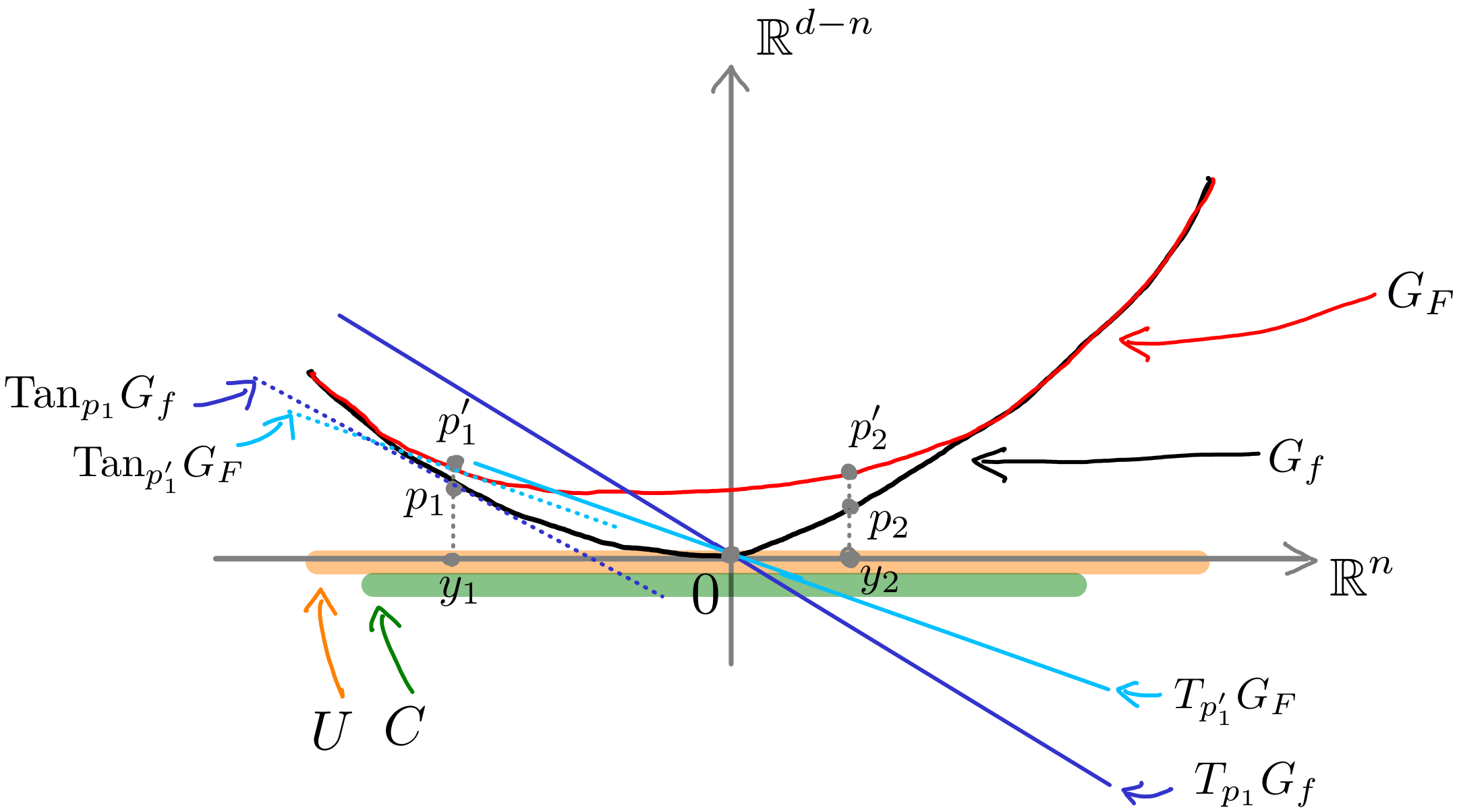}
		\caption{Illustration of the setup for this section.}
		\label{fig:setting_3.3}
	\end{figure}
	
	\noindent\textbf{We adapt the following settings, in addition to the ones established in Section~\ref{sec:interpolation_and_universal_settings}:}
	We fix two points $y_1,y_2\in C, y_1\neq y_2$, and label their graphs by
	\[p_1 = (y_1, f(y_1)), \quad p_2 = (y_2, f(y_2)) \qquad \text{and} \qquad p'_1 = (y_1, F(y_1)), \quad p'_2 = (y_2, F(y_2)).\]
	We write $G_f$ and $G_F$ for the graph of $f$ and $F$, respectively, and $T_{p_1}G_f$ and $\Tan_{p_1}G_f$ (resp. $T_{p_1'}G_F$ and $\Tan_{p_1'}G_F$) for the tangent and affine tangent space of $f$ at $p_1$ (resp. of $F$ at $p_1'$).
	
	We illustrate this setup in Figure~\ref{fig:setting_3.3}.
	\begin{lemma} \label{lem:ErrorLipschitzDer}
		The distance between the point $p_2'$ and the affine tangent space $\Tan_{p_1'}G_F$ is bounded by
		\begin{align}\label{eq:reach-tangent-cone-compatibility}
			d(p_2' , \Tan_{p_1'}G_F)\leq \tfrac{1}{2}(3 L_{\psi,D \psi} \varepsilon   +L_{Df})\abs{p_2'-p_1'}^2.
		\end{align}
		This has the following consequence: Let $\rho>0$ be small enough that the $\rho$-neighbourhood of $y_1$ is contained in $C$, $B(y_1, \rho)\subseteq C$. Then the graph $G_F$ of the function $F$ in this neighbourhood is contained in the union of balls
		\[ 
		\bigcup_{v\in\R^n :  |v| = r\leq \rho} B \left(p_1'+\begin{pmatrix} 
			v \\ DF( y_1 ) v
		\end{pmatrix} , \tfrac{1}{2}(3 L_{\psi,D \psi} \varepsilon   +L_{Df}) \:r^2 \right).
		\]  
	\end{lemma}
	We illustrate the settings of Lemma~\ref{lem:ErrorLipschitzDer} in Figure~\ref{fig:lemma_24}.
	
	\begin{figure}[h!]
		\centering
		\includegraphics[width=.65\textwidth]{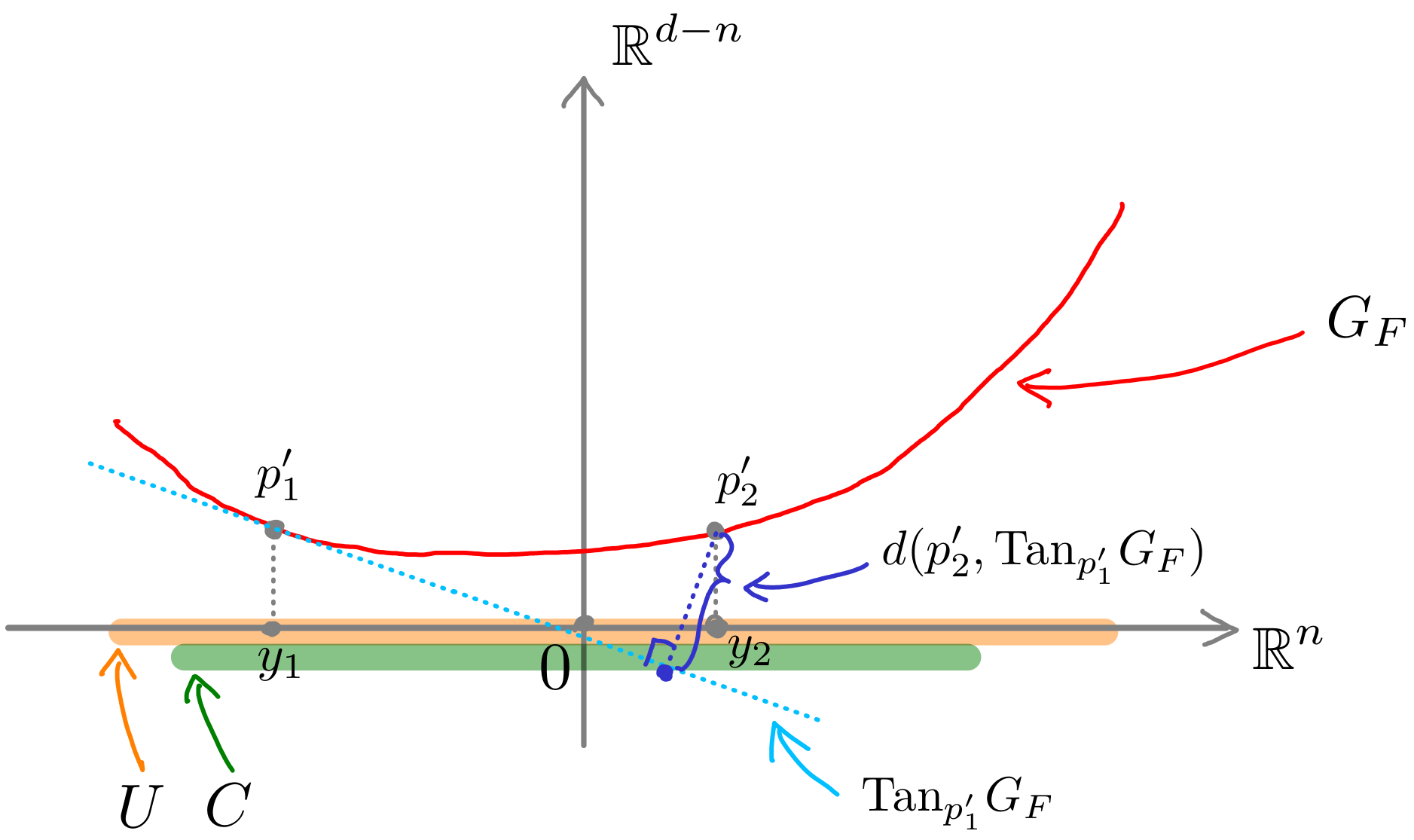}
		\caption{Illustration of the settings of Lemma~\ref{lem:ErrorLipschitzDer}.
		}
		\label{fig:lemma_24}
	\end{figure}
	
	As an auxiliary result needed to prove Proposition~\ref{prop:pointsFarAway}, we obtain a bound on the angle between the affine tangent spaces $\Tan_{p_1} G_f$ and $\Tan_{p_1'} G_F$, as well as a bound on the Hausdorff distance between two neighbourhoods contained in these spaces.
	We note that the angle between two affine spaces is equal to the angle between the corresponding vector spaces. The angle between two vector subspaces $A$ and $B$ is defined as
	\begin{align*}
		\angle A, B \defunder{=} \max_{a\in A\setminus \{0\}} \min_{b \in B \setminus \{0\}}
		\angle a,b = \max_{b \in B\setminus \{0\}} \min_{a \in A \setminus \{0\}} \angle a,b.
	\end{align*} 
	\begin{lemma}\label{boundangleTangentSpace} 
		Assume that $\varepsilon\leq\tfrac{1}{L_{\psi,D \psi}+1}$. Then the angle between the affine tangent spaces $\Tan_{p_1} G_f$ and $\Tan_{p'_1} G_F$  is bounded by
		\[
		\angle (\Tan_{p_1}G_f , \Tan_{p'_1} G_F) \leq \arcsin\left( L_{\psi,D \psi}  \varepsilon  +  \varepsilon\right).
		\] 
	\end{lemma}
	 
	\begin{figure}[h!]
		\centering
		\includegraphics[width=.5\textwidth]{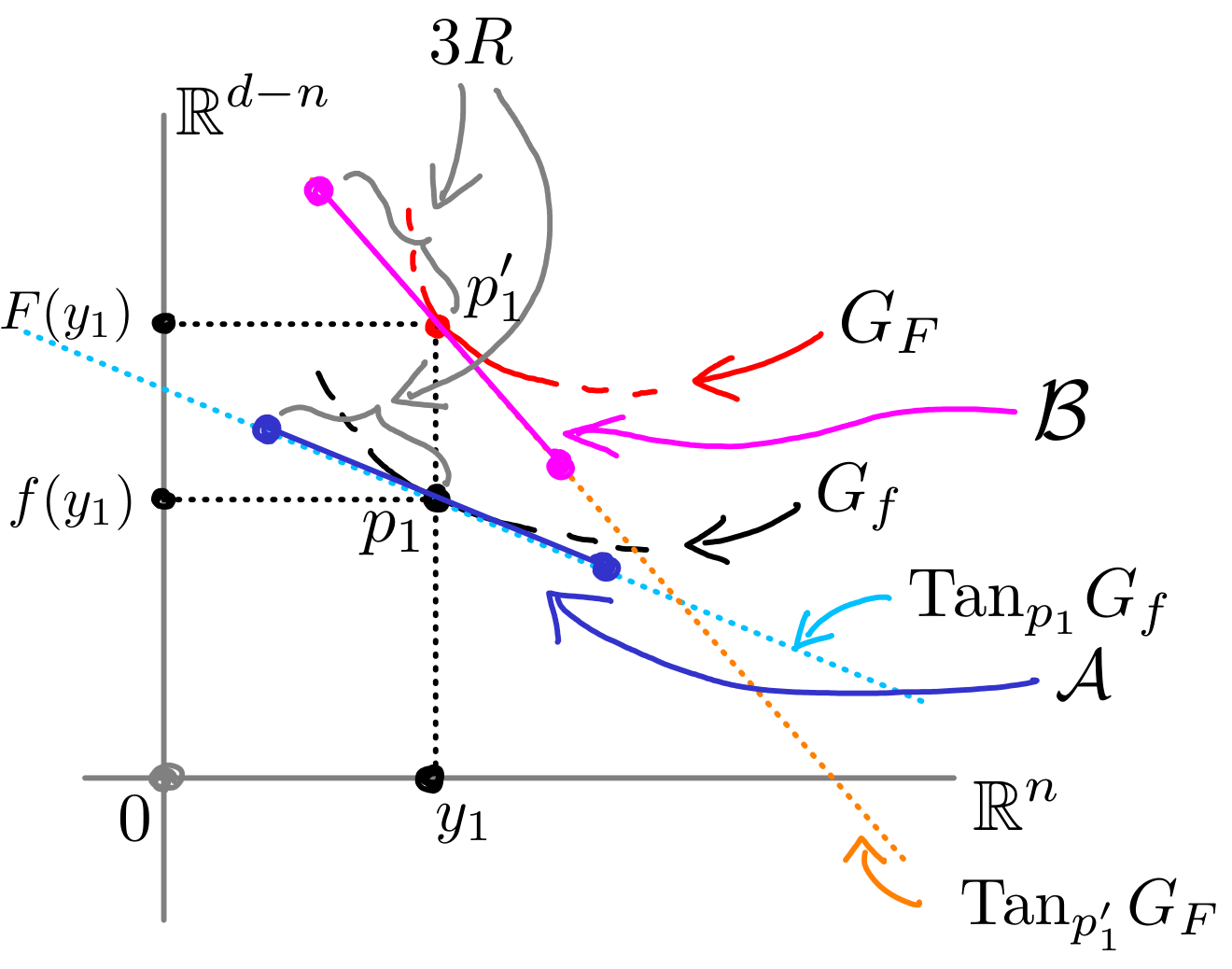}
		\caption{The 3$R$-neighbourhoods of the points $p_1$ (in dark blue) and $p_1'$ (in pink) in the corresponding affine tangent spaces.}
		\label{fig:cor_27}
	\end{figure}
	\begin{corollary}\label{CorHD}
		Let $R>0$ denote the reach of the manifold $\M$, and consider the $n$-dimensional neighbourhood of size $3R$ of the point $p_1$ in the affine tangent space $\Tan_{p_1}G_f$: 
		\begin{align*}
 {\mathcal{A} = p_1 + \{T \in T_{p_1}G_f \mid |T| \leq 3R \}.}
		\end{align*}
		Similarly, consider a $3R$-neighbourhood of the point $p_1'$ in the affine tangent space $\Tan_{p_1'}G_F$:
		\begin{align*}
 {\mathcal{B} = p_1' + \{T \in T_{p_1'}G_F \mid |T| \leq 3R \}.}
		\end{align*}
		Assume that $\varepsilon\leq\tfrac{1}{L_{\psi,D \psi}+1}$. Then the Hausdorff distance between the two neighbourhoods is upper bounded: 
		\[ 
		d_H {(\mathcal{A}, \mathcal{B})}\leq \varepsilon\left( 6R L_{\psi,D \psi}+6R+1 \right).
		\] 
	\end{corollary} 
	The settings of Corollary~\ref{CorHD} are illustrated in Figure~\ref{fig:cor_27}.

	At last, we bound the distance between the affine tangent space $p+\Tan_{p_1'} G_F$ and a point $q'\in\M'$ that does not lie on the graph of $p+F$. We illustrate the settings in Figure~\ref{fig:proposition}.
	\begin{figure}[h!]
		\centering
		\includegraphics[width=\textwidth]{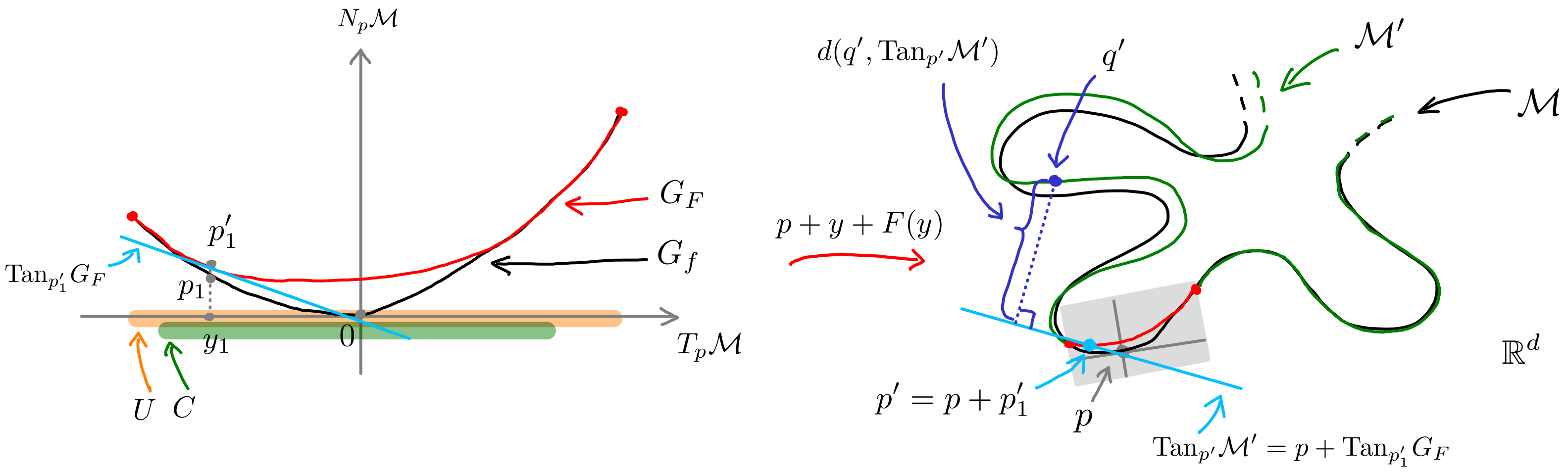}
		\caption{Illustration of the setup of Proposition~\ref{prop:pointsFarAway}.}
		\label{fig:proposition}
	\end{figure}
	\begin{proposition} \label{prop:pointsFarAway} 
		Let $q'\in\M'\backslash(p+G_F)$, and assume that $\varepsilon\leq\tfrac{1}{L_{\psi,D \psi}+1}$.
Write $p+p_1' = p'$. Then the distance between the point $q'$ and the affine tangent space $\Tan_{p'}\M' = p+\Tan_{p_1'}G_F$ is bounded by
		\begin{align}\label{eq:distance_ala_Federer_with_epsilons}
			d(q' , \Tan_{p'}\M')  \leq  \frac{|q'-p'|^2}{2R} + \frac{\varepsilon^2}{2R}   + \varepsilon\left(6R L_{\psi,D \psi}+6R+4\right).
		\end{align}
		Therefore, if $\varepsilon<R$ and $\abs{q'-p'}\geq \beta$, then 
		\[  d(q' , \Tan_{p'} \M')   \leq  \frac{|q'-p'|^2}{2R'} ,
		\] 
		with 
		\begin{align}  R'= 
\frac{R}{1+\frac{12\varepsilon R}{\beta^2}\left(R L_{\psi,D \psi}+R+1\right)}.
			\nonumber
		\end{align} 
	\end{proposition}

	\section{Proof of the main theorem} \label{sec:proof_main_theorem}
	
	We now have all the necessary tools to prove our main result. Before proceeding, let us restate the theorem:
	\MainTheorem*
	
	\begin{proof}
		This proof brings together all results from the previous sections. Roughly speaking, we split the manifold $\M$ into neighbourhoods, which we smooth iteratively one by one. Each smoothing process decreases the reach of the manifold; after each iteration, we check that the reach has not been altered `too much', so that after visiting all the neighbourhoods, we can ensure that the reach has not been decreased by more than $\varepsilon$. Let us dive in.
		\subparagraph{Step 1.}	
		In the first step we select a sample $\P$ of the manifold $\M$ (see also Figure~\ref{fig:main_proof}, right). To this end, we choose a parameter $0<\delta\leq R/2$, and let $\P$ be a $\sqrt{\delta R}/16$-net\footnote{The factor $16$ in $\sqrt{\delta R}/16$ is almost certainly suboptimal. However, it simplifies the proof by a margin.
		} on $\M$, meaning that
		\begin{itemize}
			\item for every point $x \in \M$ one can find a point $\p \in \P$, such that $|x-\p| \leq \sqrt{\delta R}/16$, 
			\item for all $\p,\q\in\P$, $|\p-\q| \geq \sqrt{\delta R}/16$.
		\end{itemize}
		Such a net exists thanks to \cite[Lemma 5.2, Section 5.1.1]{JDbook}.
		
		Due to these properties, the $\sqrt{\delta R}/16$-balls centred at the points of the sample $\P$ cover $\M$. In addition, we consider balls of radius $\sqrt{\delta R}/2$ centred at the points of $\P$, and write $N_C$ for the maximal number of such balls that intersect in {a given ball of  a ball of radius $\sqrt{\delta R}/2$}.
		By a standard packing argument (see e.g. 
\cite[Lemma 5.3, Section 5.1.1]{JDbook}), $N_C \leq \mathcal{O}(8^d)$, and in fact $N_C \leq \mathcal{O}(8^n)$ where $n$ is the dimension of the manifold. 

		
		We need one more ingredient before we start iteratively visiting neighbourhoods of each point $\p\in\P$ --- a fixed partition of unity function, which we rescale at every iteration. To this end, we split the ball $B(0,4)\subseteq\R^n$ into three sets,
		\[
		V_1 = B(0, 1), \quad V_2 = \{x\in\R^n\mid 1\leq\abs{x}< 2\}, \quad V_3 = \{x\in\R^n\mid 2\leq\abs{x}\leq 4\}.
		\]
		The two sets $\{V_1\cup V_2, V_2\cup V_3\}$ then cover the ball $B(0,4)$, and we let $\psi_0$ be the partition of unity function corresponding to the set $V_1\cup V_2$. Then $\psi_0\equiv 1$ on $V_1$, and $\psi_0\equiv 0$ on $V_3$.	
		
		Next, we select a point $\p\in\P$, and restrict our attention to the ball $B(\p,\sqrt{\delta R}/2)$ of radius $\sqrt{\delta R}/2$ centred at $\p$.
		\subparagraph{Step 2.}	
		Let $f$ be a function whose graph describes the manifold $\M$ in the ball $B(\p,\sqrt{\delta R}/2)$. To be more concrete, choose $f:T_\p\M\cap B(0,\sqrt{\delta R}/2)\to N_\p\M $ as in Theorem~\ref{lemma:embeddedManifoldPositiveReachThenC1Embedded}. Due to this theorem and our choice of $\delta$, the derivative of $f$ is $\frac{1}{R-\delta}$-Lipschitz on the whole domain $T_\p\M\cap B(0,\sqrt{\delta R}/2)$.
		
		\begin{figure}[h!]
			\centering
			\includegraphics[width=\textwidth]{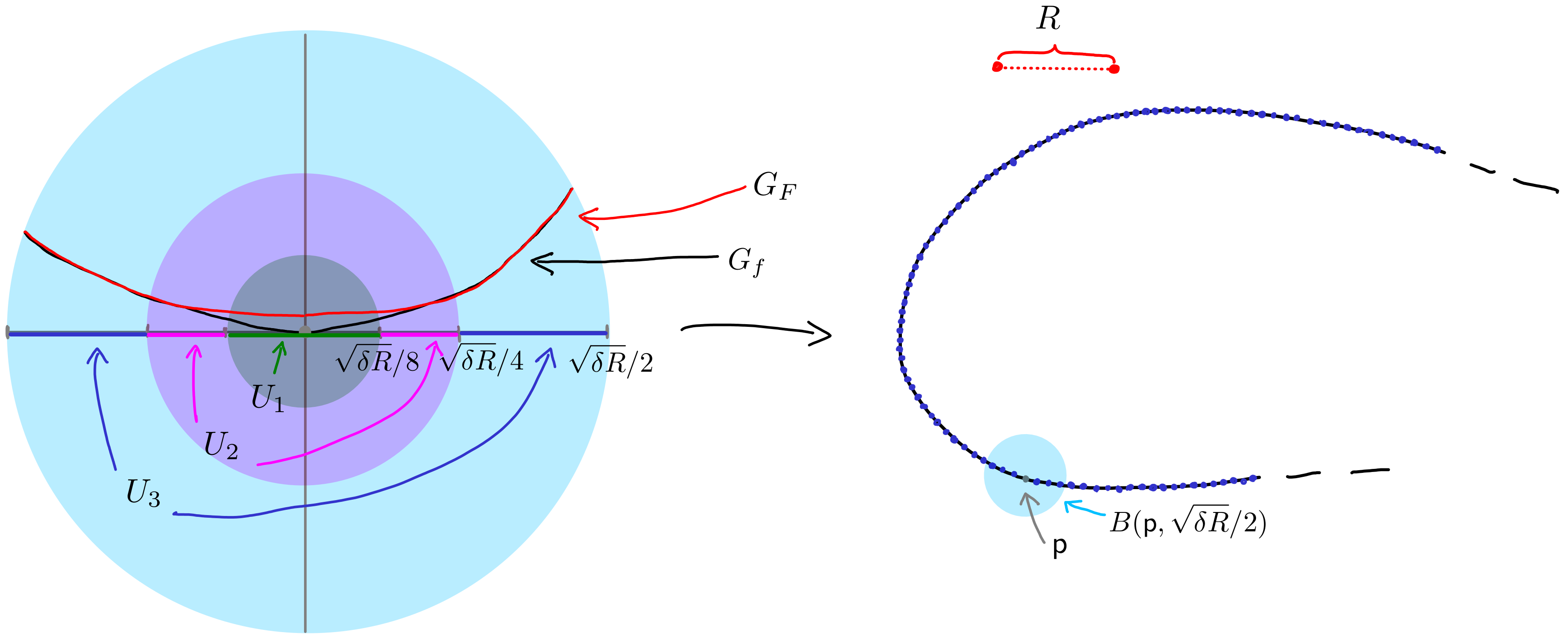}
			\caption{The point sample $\P$ at a piece of the manifold $\M$, and the neighbourhoods $U_1, U_2$, and $U_3$.}
			\label{fig:main_proof}
		\end{figure}
		
		We apply the smoothing construction described in Section~\ref{sec:interpolation_and_universal_settings}. As ingredients, we need a partition of unity function, and a smoothing kernel.
		
		\noindent\emph{Partition of unity function:} We map the neighbourhood $T_\p\M\cap B(0,\sqrt{\delta R}/2)$ diffeomorphically to the ball $B(0,4)\subseteq \R^n$ using the canonical identification of the tangent space $T_\p\M$ with $\R^n$ and the map $x\mapsto \frac{8}{\sqrt{\delta R}}\cdot x$. The preimages of the sets $V_1, V_2$, and $V_3$ under this map are, respectively,
		\begin{align*}
			&U_1 = \left\{ x\in T_\p\M \mid \abs{x}<\sqrt{\delta R}/8 \right\}, \\
			&U_2 = \left\{ x\in T_\p\M \mid \sqrt{\delta R}/8\leq \abs{x}<\sqrt{\delta R}/4 \right\}, \\
			&U_3 = \left\{ x\in T_\p\M \mid \sqrt{\delta R}/4\leq \abs{x}\leq\sqrt{\delta R}/2 \right\}.
		\end{align*}
		The sets $U_1, U_2,$ and $U_3$ are illustrated in Figure~\ref{fig:main_proof}, on the left.
		We define $\psi$ by {scaling} the partition of unity function $\psi_0$: $\psi(x) = \psi_0 \left(\frac{8x}{\sqrt{\delta R}}\right )$. Then $\psi$ is the partition of unity function corresponding to the set $U_1\cup U_2$ in the cover $\{U_1\cup U_2, U_2\cup U_3\}$ of $T_\p\M\cap B(0,\sqrt{\delta R}/2)$, and thereby, $\psi\equiv 1$ on $U_1$, and $\psi\equiv 0$ on $U_3$. Furthermore:
		

		\begin{lemma}[Folklore]\label{lemma:Lipschitz_conts_psi}
			The Lipschitz constants of $\psi$ and its derivative are bounded by
			\begin{align*}
				L_{\psi}\leq\frac{8}{\sqrt{\delta R}} L_{\psi_0} \qquad \text{and} \qquad L_{D\psi}\leq\frac{64}{\delta R} L_{D\psi_0}.
			\end{align*}
		\end{lemma}

		By decreasing $\delta$ if necessary, we can also assume that $\frac{8}{\sqrt{\delta R}} L_{\psi_0}\leq \frac{64}{\delta R} L_{D\psi_0}$. Thus:
		\begin{align}\label{eq:Lipschitz_const_final_psi}
			L_{\psi, D\psi} = \max \{L_\psi, L_{D\psi} \}\leq \frac{64}{\delta R} L_{D\psi_0}.
		\end{align}

		\noindent\emph{Smoothing kernel:} Let $C = \overline{\supp \psi}$. For $0<\rho<\frac{1}{1+\frac{64}{\delta R} L_{D\psi_0}}$ we choose a smoothing kernel $\varphi_\rho:T_\p\M\cap B(0,\sqrt{\delta R}/2)\to \R$  such that
		\begin{align} 
			\| \varphi_{\rho}  \ast f - f \| _{1, C}  \leq \rho.
			\tag{\ref{eq:AssumptionSmoothing1}}
		\end{align} 
		Such kernel exists due to Theorem \ref{Th:Hirsch23c}.
		With these choices we define 
		\begin{align}
			F:T_\p\M\cap B(0,\sqrt{\delta R}/2)\to N_\p\M, \qquad x\mapsto F(x)= (1- \psi(x))  f (x) + \psi (x)(\varphi_\rho \ast f(x)).\tag{\ref{eq:smoothed_mfld_function}}
		\end{align}

		\subparagraph{Step 3.}	
		We write $\M'$ for the manifold that coincides with $\M$ outside of the ball of radius $\sqrt{\delta R}/2$ centred at the point $p$, and is the graph of $\p+F$ inside this ball.
		
		To be more concrete, we recall that we can write each point $p\in\M\cap B(\p, \sqrt{\delta R}/2)$ as $p = \p + y + f(y)$, with $y\in T_\p\M$. With this in mind we define, for each $p\in\M$:
		\begin{align*}
			p' \defunder{=}  \begin{cases}
				\p + y + F(y), & \text{if } p\in B(\p, \sqrt{\delta R}/2),\\
				p, & \text{otherwise},
			\end{cases}
		\end{align*}
	and let $\M' = \{p'\mid p\in\M\}$.
%
		We stress that
		\begin{itemize}
			\item inside the ball $B(\p, \sqrt{\delta R}/16)$, $\M'$ is smooth ($C^{\infty}$), and
			\item $\M'=\M$ not only outside of the ball $B(\p, \sqrt{\delta R}/2)$, but already outside of the ball $B(\p, \sqrt{\delta R}/4)$.
		\end{itemize} 
		
		Next, we set forth to bound the reach of $\M'$. 
		We choose two points $p',q'\in\M'$, and estimate the distance between $q'$ and the affine tangent space $\Tan_{p'}\M'$ --- our bounds from Lemma~\ref{lem:ErrorLipschitzDer} and Proposition~\ref{prop:pointsFarAway} yield:
		\begin{lemma}\label{lemma:final_approximation_distance_between_q'_and_tangent_space_of_p'}
			Let $p',q'\in\M'$ be two points in the manifold~$\M'$. Then the distance between $q'$ and the affine tangent space $\Tan_{p'}\M'$ is bounded by
			\begin{align*}
				d(q', \Tan_{p'}\M')\leq \frac{|p'-q'|^2}{2R'},
			\end{align*}
		with 
		\begin{align}\label{eq:final_approximation_distance_between_q'_and_tangent_space_of_p'}
			R' = \frac{R}{\max\left\{3\cdot 64 L_{D\psi_0}\tfrac{\rho}{\delta} + \tfrac{1}{1-\tfrac{\delta}{R}}, 1+ 3\cdot 64\tfrac{\rho}{\delta}\left(\tfrac{64}{\delta}L_{D\psi_0} + R+ 1\right) \right\}}.  
		\end{align}
		\end{lemma}
		Thus, due to Theorem~\ref{th418Federer}, the reach of the manifold $\M'$ satisfies $\rch(\M')\geq R',$ with $R'$ as in equation~\eqref{eq:final_approximation_distance_between_q'_and_tangent_space_of_p'}.
		
		\subparagraph{Step 4.}	
		The final step is again subdivided into two stages: First, we observe that \eqref{eq:final_approximation_distance_between_q'_and_tangent_space_of_p'} can be made as close to $R$ as needed. In the second stage we exploit the fact that the definition of the reach is local in $\M \times \M$, as is clear from Theorem \ref{th418Federer}, that is we consider only neighbourhoods of $p$ and $q$.

		For the given $\varepsilon>0$, we now choose $\rho$ and $\delta$ such that $R'$ from equation~\eqref{eq:final_approximation_distance_between_q'_and_tangent_space_of_p'} satisfies $| R -R'  |  \leq \varepsilon/ (2 N_C)$.
		This is always possible. One can set $\delta =  \varepsilon/ (4 N_C)$; the choice for $\rho$ is a bit more subtle but choosing $\rho = \mathcal{O} (\delta^4)$ works if $\varepsilon= \mathcal{O} (\delta)$ is sufficiently small.

	For a given $p, q \in \M$ we write $\M_{p,q}$ for the manifold that is smoothed in the neighbourhoods of $p$ and $q$. That is, we choose the points $\p_0,\mathsf{q}_0\in \P$ that are the closest to $p$ and $q$, respectively. 
	We then consider all points $\p'\in \P$ for which the ball $B(\p', \sqrt{\delta R}/2)$ intersects either the ball $B(\p_0, \sqrt{\delta R}/2)$ or the ball $B(\mathsf{q}_0, \sqrt{\delta R}/2)$. We call the set of such points $\P'\subseteq \P$. 

	There are at most $2N_C$ points in $\P'$. We stress that $N_C$ is a constant, depending only on the dimension of the manifold $\M$. Next, we apply our smoothing construction to all points in $\P'$, and call the resulting manifold $\M'_{p,q}$. 
	Formally speaking, the points in $\P'$ may shift `a little' by the construction. However, we can assume without loss of generality (by choosing $\varepsilon$ small enough) that this shift is never more than $\sqrt{\delta R}/32$. 
This means that the covering by disks of radius $\sqrt{\delta R} / 8$ (centred at the sample points of $\P$) of the manifold is preserved after the shift, in particular every point in the manifold is smoothed. 	

For all points of $ \M'_{p,q}$ and in particular for the images of the $p,q$ we started out with, denoted $p',q' \in \M'_{p,q}$, we have that 
			\begin{align}
				d(q', \Tan_{p'}\M'_{p,q})\leq \frac{|p'-q'|^2}{2R'},
				\label{eq:COnclusionIteration} 
			\end{align}
with $| R -R'  |  \leq \varepsilon$. 
	

Now note that when we perform our construction iteratively on all points in $\P$ and call the result $\M'_{\textrm{final}}$, the neighbourhoods of the points $p'$ and $q'$ no longer change. This means that \eqref{eq:COnclusionIteration} holds for $\M'_{\textrm{final}}$ as well. And because $p,q$ were chosen arbitrarily, it holds for any pair of $p',q'\in \M'_{\textrm{final}}$. In other words, the manifold $\M'_{\textrm{final}}$ has reach at least $R-\varepsilon$.
%
		%
		%
		%
Moreover, the distance between the manifold before and after the smoothing is governed by the distance between the graphs of the functions $f$ and $F$, which, due to inequality~\eqref{eq:AssumptionSmoothing1}, is bound by $\varepsilon/N_C$.
	\end{proof}

	\bibliography{geomrefs}
	
	\appendix
	\section{Proofs} \label{sec:proofs}
	\begin{proof}[Proof of Lemma~\ref{lemma:norm_of_smoothing_depends_on_Lipschitz_const_only}] We have that, 
	{
\allowdisplaybreaks
		\begin{align}   
			\| \varphi \ast f - f \| _{0, C} &= \sup_{x\in U_\sigma}\left | \int_{\mathbb{R}^n} \varphi (y) f(x-y) \ud y - f (x) \right | 
			\nonumber 
			\\ 
			&= \sup_{x\in U_\sigma}\left | \int_{\mathbb{R}^n} \varphi (y) f(x-y) -  \varphi (y) f (x)  \ud y  \right | 
			\tag{because $\int \varphi = 1$} 
			\\ 
			& \leq \sup_{x\in U_\sigma} \int_{\mathbb{R}^n} \left | \varphi (y) f(x-y) -  \varphi (y) f (x)  \right |  \ud y  
			\tag{by the triangle inequality for integrals}
			\\
			& \leq \sup_{x\in U_\sigma} \int_{\mathbb{R}^n}  \varphi (y)  \left |f(x-y) -  f (x)  \right |  \ud y  
			\nonumber
			\\
			&\leq  L \sigma ,
			\tag{because $ | f(x-y) -f(x)| \leq L\sigma $ and  $\int \varphi = 1$}
		\end{align}  }
	\end{proof}
	
	
	\begin{proof}[Proof of Lemma \ref{lem:SmoothingConservesLipschitzConstants}] 
		By definition of a smoothing kernel, $\int  \varphi(z) =1$. Therefore,
{
\allowdisplaybreaks
		\begin{align} 
			\left |\varphi \ast g (y_2) -  \varphi \ast g(y_1) \right| &= \left| \int_z \varphi (z) g(y_2-z) \ud z  -  \int_z \varphi (z) g(y_1-z) \ud z \right | 
			\nonumber 
			\\
			&=  \left | \int_z \varphi (z) (g(y_2-z)  - g(y_1-z)) \ud z \right | 
			\nonumber
			\\
			&\leq   \int_z \varphi (z) |(g(y_2-z)  - g(y_1-z)) |\ud z  
			\tag{because $\varphi$ is non-negative}
			\\
			& \leq  L | y_2-y_1| \int _z \varphi (z) \ud z
			\tag{by the definition of the Lipschitz constant} 
			\\
			& = L  | y_2-y_1|.  
			\tag{because $\int  \varphi(z) =1$} 
		\end{align}}
	\end{proof} 
	

	\begin{proof}[Proof of Lemma \ref{lem:smoothingPreservesLipschitzConstantDerivative}] 
		Let $v\in \mathbb{R}^{d-n}$ and $w\in \mathbb{R}^n$ be two vectors with $|v|=|w|=1$.

		Because for every pair of points $y_1,y_2\in\R^n$ holds $\big\|Dg(y_2)  - Dg(y_1)  \big\|_{2}  \leq L | y_2 - y_1 |$, the function $y \mapsto \langle v , D(g(y))(w) \rangle$ is $L$-Lipschitz. In other words,
		\[\langle v , D(g(y_2) )(w) \rangle   -\langle v ,   D(g(y_1))(w)  \rangle = \langle v , (D(g(y_2))-D(g(y_1)))(w) \rangle  \leq L |y_2 - y_1| .\]
		Therefore,
		\begin{align} 
			\langle v ,D (\varphi \ast g(y_2) )(w) \rangle   -\langle v , D (\varphi \ast g(y_1))(w)  \rangle  &= \langle v , \varphi \ast D(g(y_2) )(w) \rangle   -\langle v , \varphi \ast  D(g(y_1))(w)  \rangle 
			\tag{because $g$ is $C^1$}  
			\\
			&=  \varphi \ast \langle v , D(g(y_2) )(w) \rangle   - \varphi \ast \langle v ,   D(g(y_1))(w)  \rangle 
			\tag{by exchanging the order of the integral and the inner product}
			\\ 
			&\leq  L |y_2 - y_1|  .
			\tag{by Lemma \ref{lem:SmoothingConservesLipschitzConstants}} 
		\end{align} 
		Thus, by definition of the operator norm,
		\[ \big\|D (\varphi \ast g(y_2))  - D(\varphi \ast g(y_1))   \big\|_{2}  \leq L | y_2 - y_1 |.\]
	\end{proof} 
	
	\begin{proof}[Proof of Lemma \ref{Lem:LipschitzWithoutDer}]
		For $y_1,y_2\in\R^n$,
		
\begingroup
\allowdisplaybreaks
		\begin{align}
			| F(\yx) & -F(\yy) | 
			\nonumber 
			\\&= \big| (1- \psi(\yx) )  f (\yx) + \psi (\yx)(\varphi_{\varepsilon} \ast f(\yx)) - (1- \psi(\yy) )  f (\yy) - \psi (\yy)(\varphi_{\varepsilon} \ast f(\yy)) \big| 
			\nonumber 
			\\
			&= 
			\big |  \psi (\yx) \big(\varphi_{\varepsilon} \ast f(\yx)- f (\yx)\big)  - \psi (\yy) \big(\varphi_{\varepsilon} \ast f(\yy) -  f (\yy) \big) + f (\yx)- f (\yy) \big| 
			\tag{by shuffling terms}
			\\
			&= 
			\big |  \psi (\yx) \big(\varphi_{\varepsilon} \ast f(\yx)- f (\yx)\big) + f (\yx)- f (\yy) 
			\nonumber
			\\ 
			& \phantom{=} - \psi (\yy) \big(\varphi_{\varepsilon} \ast f(\yy) -  f (\yy) -\varphi_{\varepsilon} \ast f(\yx)+ f (\yx)+ \varphi_{\varepsilon} \ast f(\yx)- f (\yx)\big)  \big| 
			\tag{by adding $0= \psi(y_1)(-\varphi_{\varepsilon} \ast f(\yx)+ f (\yx)+ \varphi_{\varepsilon} \ast f(\yx)- f (\yx))$ }
			\nonumber 
			\\
			&= 
			\big |  (\psi (\yx)- \psi (\yy)) \big(\varphi_{\varepsilon} \ast f(\yx)- f (\yx)\big) + f (\yx)- f (\yy) 
			\nonumber
			\\ 
			& \phantom{=} - \psi (\yy) \big(\varphi_{\varepsilon} \ast f(\yy) -  f (\yy) -\varphi_{\varepsilon} \ast f(\yx)+ f (\yx)\big)  \big| 
			\nonumber 
			\\
			&= 
			\big |  (\psi (\yx)- \psi (\yy)) \big(\varphi_{\varepsilon} \ast f(\yx)- f (\yx)\big) + (1-\psi(\yy) )(f (\yx)- f (\yy) )
			\nonumber
			\\ 
			& \phantom{=} - \psi (\yy) \big(\varphi_{\varepsilon} \ast f(\yy)  -\varphi_{\varepsilon} \ast f(\yx) \big)  \big| 
			\nonumber
			\\
			& \leq 
			\big |  (\psi (\yx)- \psi (\yy)) \big(\varphi_{\varepsilon} \ast f(\yx)- f (\yx)\big)\big |  + \big | (1-\psi(\yy) )(f (\yx)- f (\yy) )\big | 
			\nonumber
			\\ 
			& \phantom{=}+ \big |  \psi (\yy) \big(\varphi_{\varepsilon} \ast f(\yy)  -\varphi_{\varepsilon} \ast f(\yx) \big)  \big| 
			\tag{by the triangle inequality}
			\\
			& = 
			\big |  \psi (\yx)- \psi (\yy) \big |  \cdot \big | \varphi_{\varepsilon} \ast f(\yx)- f (\yx)\big |  + \big | 1-\psi(\yy) \big | \cdot \big |f (\yx)- f (\yy) \big | 
			\nonumber
			\\ 
			& \phantom{=}+ \big |  \psi (\yy) \big | \cdot  \big |   \varphi_{\varepsilon} \ast f(\yy)  -\varphi_{\varepsilon} \ast f(\yx)   \big| 
			\tag{because $\psi$ is a scalar}
			\\
			& \leq
			L_\psi |\yx-\yy|   \cdot \varepsilon  + \big | 1-\psi(\yy) \big | \cdot L |\yx-\yy| 
			\tag{by the Lipschitz assumptions and because $\| \varphi \ast f - f \| _{k, C} \leq \varepsilon$}
			\\ 
			& \phantom{=}+ \big |  \psi (\yy) \big | \cdot L |\yx-\yy|
			\tag{by Lemma \ref{lem:SmoothingConservesLipschitzConstants}}
			\\
			&= \varepsilon L_\psi    |\yx-\yy|  +  L |\yx-\yy| .
			\tag{because $\psi (\yy)\in [0,1]$}
		\end{align}
		\endgroup
	\end{proof}
	
	\begin{proof}[Proof of Lemma~\ref{lem:SmoothingWithPartition}]
		Choose $y_1,y_2\in C$. We first consider 
		\[g(x) =  \varphi_\varepsilon \ast f(x) - f(x).\]
		With $v=\frac{\yy-\yx}{|\yy-\yx|}$, the line segment from $y_2$ to $y_1$ is parameterized as
		\[\gamma:[0,\abs{y_2-y_1}]\to \R^n, \qquad \gamma(t)= y_2+t\cdot v.\]
		
		Then, due to the fundamental theorem of calculus,
				\begingroup
\allowdisplaybreaks
		\begin{align}
			\abs{g(y_2)-g(y_1)}&= \abs{\int_{0}^{|\yx-\yy|} \tfrac{\ud}{\ud t} g(\gamma(t)) \ud t}
			\nonumber
			\\&= \abs{\int_{0}^{|\yx-\yy|} \left((D(\varphi \ast f) - Df) \left(\yx+ v t\right)\right) \left( v \right)\ud t}
			\nonumber
			\\&= \abs{\int_{0}^{|\yx-\yy|} \left(((\varphi \ast Df) - Df) \left(\yx+ v t\right)\right) \left( v \right)\ud t}
			\tag{by Theorem \ref{Th23abHirsch}}
			\\&\leq\int_{0}^{|\yx-\yy|} \abs{\left(((\varphi \ast Df) - Df) \left(\yx+ v t\right)\right) \left( v \right)}\ud t 
			\tag{by the triangle inequality for integrals} 
			\\&\leq \int_{0}^{|\yx-\yy|} \varepsilon |v| \ud t
			\tag{by \eqref{eq:AssumptionSmoothing2}}
			\\&= \varepsilon |\yx-\yy|. 	\label{DifferenceSmoothings}
		\end{align}
		\endgroup
		
		With this, we bound the Lipschitz constant of the derivative $DF$:
		\begingroup
\allowdisplaybreaks
		\begin{align} 
			&\big\|DF(y_2)  - DF(y_1)  \big\|_{2} 
			\nonumber
			\\
			&= \big\| D \left( (1- \psi(y_2) )  f (y_2)+ \psi (y_2) (\varphi_{\varepsilon} \ast f (y_2)) \right) 
			\nonumber
			\\ 
			&\phantom{=\big\| } -
			D \left( (1- \psi(y_1) )  f (y_1)+ \psi (y_1) (\varphi_{\varepsilon} \ast f (y_1)) \right) \big\|_{2}  
			\nonumber 
			\\
			&= \big\| 
			D(\psi(y_2)) ( \underbrace{\varphi_{\varepsilon} \ast f (y_2)  -  f (y_2)}_{g(y_2)} )- D(\psi(y_1))  ( \underbrace{\varphi_{\varepsilon} \ast f (y_1) -f (y_1)}_{g(y_1)} ) 
			\nonumber 
			\\ & \phantom{= \big\| \,} 
			+   \psi(y_2)   (\underbrace{\varphi_{\varepsilon} \ast Df (y_2)   -D f (y_2)}_{Dg(y_2)})-  \psi(y_1) (\underbrace{\varphi_{\varepsilon} \ast Df (y_1)   - D f (y_1)}_{Dg(y_1)})
			\nonumber 
			\\ & \phantom{= \big\| \,} 
			+   D f (y_2)-  D f (y_1)
			\big\|_{2}  
			\tag{by Theorem \ref{Th23abHirsch} and reshuffling} 
			\\
			&= \big\| 
			D\psi(y_2)g(y_2)
			- D\psi(y_1)  ( g(y_1)  + g(y_2)-g(y_2) ) 
			\tag{by inserting $0= g(y_2)-g(y_2)$}
			\\ & \phantom{= \big\| \,} 
			+   \psi(y_2)  Dg(y_2)
			-  \psi(y_1) (  Dg(y_1) +Dg(y_2)-Dg(y_2))
			\tag{by inserting $0=Dg(y_2)-Dg(y_2)$}
			\\ & \phantom{= \big\| \,} 
			+   D f (y_2)-  D f (y_1)
			\big\|_{2}  
			\nonumber
			\\
			&= \big\| 
			(D\psi(y_2)- D\psi(y_1) ) g(y_2)
			- D\psi(y_1)  ( g(y_1) -g(y_2) ) 
			\nonumber
			\\ & \phantom{= \big\| \,} 
			+   (\psi(y_2)-   \psi(y_1))  Dg(y_2)
			-  \psi(y_1) (  \underbrace{Dg(y_1)}_{\varphi_{\varepsilon} \ast Df (y_1)   - D f (y_1)} -\underbrace{Dg(y_2)}_{\varphi_{\varepsilon} \ast Df (y_2)   - D f (y_2)})
			\nonumber
			\\ & \phantom{= \big\| \,} 
			+   D f (y_2)-  D f (y_1)
			\big\|_{2}  
			\nonumber
			\\
			&\leq  \big\| 
			(D\psi(y_2)- D\psi(y_1))\big\|_{2}\cdot\abs{g(y_2)}  
			+\big\|  D(\psi(y_1))\big\|_{2} \cdot \abs{g(y_1)-g(y_2)}   
			\nonumber
			\\ & \phantom{= \big\| \,} 
			+  \abs{\psi(y_2)-   \psi(y_1)}\cdot\big\|Dg(y_2)\big\|_{2}  
			+   \psi(y_1) \cdot\big\|  \varphi_{\varepsilon} \ast Df (y_1)    -\varphi_{\varepsilon} \ast Df (y_2)  \big\|_{2}  
			\nonumber
			\\ & \phantom{= \big\| \,} 
			+  (1- \psi(y_1)) \cdot \big\| D f (y_2)-  D f (y_1)
			\big\|_{2}  
			\tag{by the triangle inequality}
			\\
			&\leq  L_{\psi,D \psi} |y_2-y_1| \cdot \varepsilon  
			\tag{by \eqref{eq:Lipschitz_cont_derivative}, \eqref{eq:LipschitzpsiDpsi} and \eqref{eq:AssumptionSmoothing3}} 
			\\ 
			& \phantom{= \big\| \,}
			+L_{\psi,D \psi} \cdot\varepsilon |y_2-y_1|  
			\tag{by \eqref{eq:boundDpsi} and \eqref{DifferenceSmoothings}}
			\\ & \phantom{= \big\| \,} 
			+ L_{\psi,D \psi} |y_2-y_1|\cdot \varepsilon    
			\tag{by \eqref{eq:Lipschitz_cont}, \eqref{eq:LipschitzpsiDpsi} and \eqref{eq:AssumptionSmoothing2} } 
			\\ 
			& \phantom{= \big\| \,}
			+ \psi (y_1) \cdot L_{Df} |y_2-y_1| 
			\tag{because $\psi(y_1) \in [0,1]$ and by Lemma \ref{lem:smoothingPreservesLipschitzConstantDerivative}}
			\\ & \phantom{= \big\| \,} 
			+  (1-\psi  (y_1)) \cdot  L_{Df} |y_2-y_1|  
			\tag{because $\psi(y_1) \in [0,1]$ and by \eqref{eq:Lipschitz_cont_derivative}}
			\\
			&= (3 L_{\psi,D \psi} \varepsilon   +L_{Df}) \: |\yx-\yy| .
			\nonumber
		\end{align} 
		\endgroup
	\end{proof}

	\begin{proof}[Proof of Lemma~\ref{lem:ErrorLipschitzDer}]
		
		\begin{figure}[h!]
			\centering
			\includegraphics[width=0.6\textwidth]{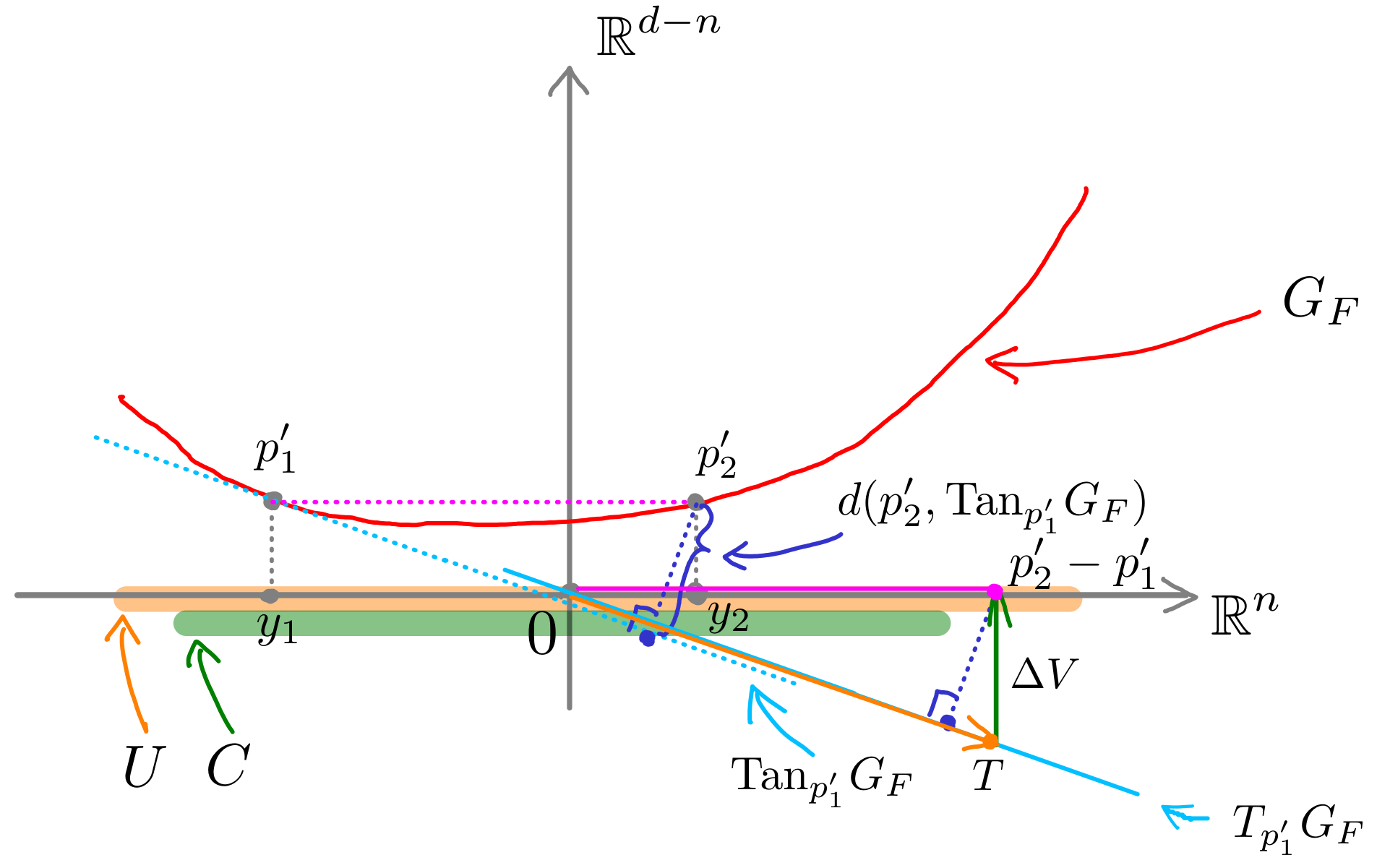}
			\caption{Notation used in the proof of Lemma~\ref{lem:ErrorLipschitzDer}.}.
			\label{fig:lemma_24_proof}
		\end{figure}
		
		We first split the vector $p_2'-p_1'$ into a sum of two vectors $T$ and $\Delta V$, such that $T\in T_{p_1'}G_F$ and $\Delta V\in 0\times \R^{d-n}$. 
		To this end, parameterize the line segment connecting $y_2$ and $y_1$ by 
		$t \in [0,1], t\mapsto y_1 + (y_2-y_1) t$.
		By the fundamental theorem of calculus,
		\begin{align} 
			&p_2'-p_1' = \begin{pmatrix} 
				y_2 \\ F(y_2)
			\end{pmatrix} 
			-
			\begin{pmatrix} 
				y_1 \\ F(y_1)
			\end{pmatrix}
			= \int_0^1 
			\begin{pmatrix} 
				y_2- y_1 \\ DF( y_1 + (y_2-y_1) t ) (y_2-y_1)
			\end{pmatrix}
			\ud t
			\nonumber 
			\\ 
			&= \int_0^1 
			\begin{pmatrix} 
				y_2- y_1 \\ DF( y_1 ) (y_2-y_1) + DF( y_1 + (y_2-y_1) t ) (y_2-y_1) - DF( y_1 ) (y_2-y_1) 
			\end{pmatrix}
			\ud t
			\tag{by adding $0= DF(y_1) (y_2-y_1) -DF(y_1) (y_2-y_1)$}
			\\ 
			&= \begin{pmatrix} 
				y_2- y_1 \\ DF( y_1 ) (y_2-y_1) 
			\end{pmatrix} 
			+ \int_0^1  
			\begin{pmatrix}
				0 \\  \left(DF( y_1 + (y_2-y_1) t )- DF( y_1 )\right) (y_2-y_1) 
			\end{pmatrix} \ud t.
			\nonumber
		\end{align} 
		We set 
		\begin{align*}
			&T= \begin{pmatrix} 
				y_2- y_1 \\ DF( y_1 ) (y_2-y_1) 
			\end{pmatrix}, &\Delta V= 
			\int_0^1
			\begin{pmatrix}
				0 \\ \left(DF( y_1 + (y_2-y_1) t )- DF( y_1 )\right) (y_2-y_1) 
			\end{pmatrix} \ud t.
		\end{align*}
		As one can see from Figure~\ref{fig:lemma_24_proof}, the distance $d(p_2', \Tan_{p_1'}G_F) = d(p_2'-p_1', T_{p_1'}G_F) $ is the height at vertex $p_2'-p_1'$ of the triangle with vertices $0, T$, and $p_2'-p_1'=T+\Delta V$, and thus
		\begin{align} 
			d(p_2', \Tan_{p_1'}G_F)  \leq |\Delta V|.
			\label{eqDistTanS1} 
		\end{align} 
		Using the Lipschitz assumption on the derivative of $F$ and the triangle inequality for integrals, we further estimate
		\begin{align} \label{eq:delta_V_estimate}
			|\Delta V| \leq  
			\int_0^1 L_{DF} |y_2-y_1|^2 t \: \ud t  =\tfrac{1}{2} L_{DF} |y_2-y_1|^2 .
		\end{align} 
		Finally, by the triangle inequality,
		\begin{align} 
			|y_2-y_1| \leq  |(y_2, F(y_2))-(y_1, F(y_1))| = \abs{p_2'-p_1'}.
			\label{TriangInequal} 
		\end{align} 
		Combining \eqref{eqDistTanS1}, \eqref{eq:delta_V_estimate}, \eqref{TriangInequal}, and the bound on the Lipschitz constant from Lemma~\ref{lem:SmoothingWithPartition} yields equation \eqref{eq:reach-tangent-cone-compatibility}:
		\begin{align*} 
			d(p_2', \Tan_{p_1'}G_F)   \leq \tfrac{1}{2} L_{DF}\abs{p_2'-p_1'}^2\leq \tfrac{1}{2} (3 L_{\psi,D \psi} \varepsilon   +L_{Df})\abs{p_2'-p_1'}^2.
		\end{align*}
		To see that $p_2'\in B \left(p_1'+T , \frac{1}{2} L_{DF} |y_1-y_2|^2 \right)$, notice that this is equivalent to $p_2'-p_1' = T+\Delta V$ being contained in the ball $B \left(T , \frac{1}{2} L_{DF} |y_1-y_2|^2 \right)$, which follows from inequality \eqref{eq:delta_V_estimate}.
	\end{proof}
	
	We need the following lemma to prove Lemma~\ref{boundangleTangentSpace}:
	\begin{lemma}\label{lemma:angle_between_tangent_spaces}
		Consider two maps $g, \tilde{g} : \mathbb{R}^n \to \mathbb{R}^{d-n}$, whose derivatives are close in a neighbourhood of a point $y\in\R^n$ --- more concretely, there exists a $\delta <1$ such that
		\[
		\| (Dg- D\tilde{g})(y)\|_2 \leq  \delta.
		\]
		Let $G_g$ and $G_{\tilde{g}}$ denote the graph of $g$ and $\tilde{g}$, respectively, and $T_{(y,g (y))}G_g $ and $ T_{(y,\tilde{g}(y))} G_{\tilde{g}}$ their respective tangent spaces at the graph of the point $y$. Then the angle between $T_{(y,g (y))}G_g $ and $ T_{(y,\tilde{g}(y))} G_{\tilde{g}}$ satisfies
		\begin{align}\label{eq:angle_between_tangent_spaces}
			\angle (T_{(y,g (y))}G_g , T_{(y,\tilde{g}(y))} G_{\tilde{g}} ) \leq \arcsin \delta.
		\end{align}
	\end{lemma}
	The settings of Lemma~\ref{lemma:angle_between_tangent_spaces} are illustrated in Figure~\ref{fig:lemma_26}.
	\begin{figure}[h!]
		\centering
		\includegraphics[width=0.5\textwidth]{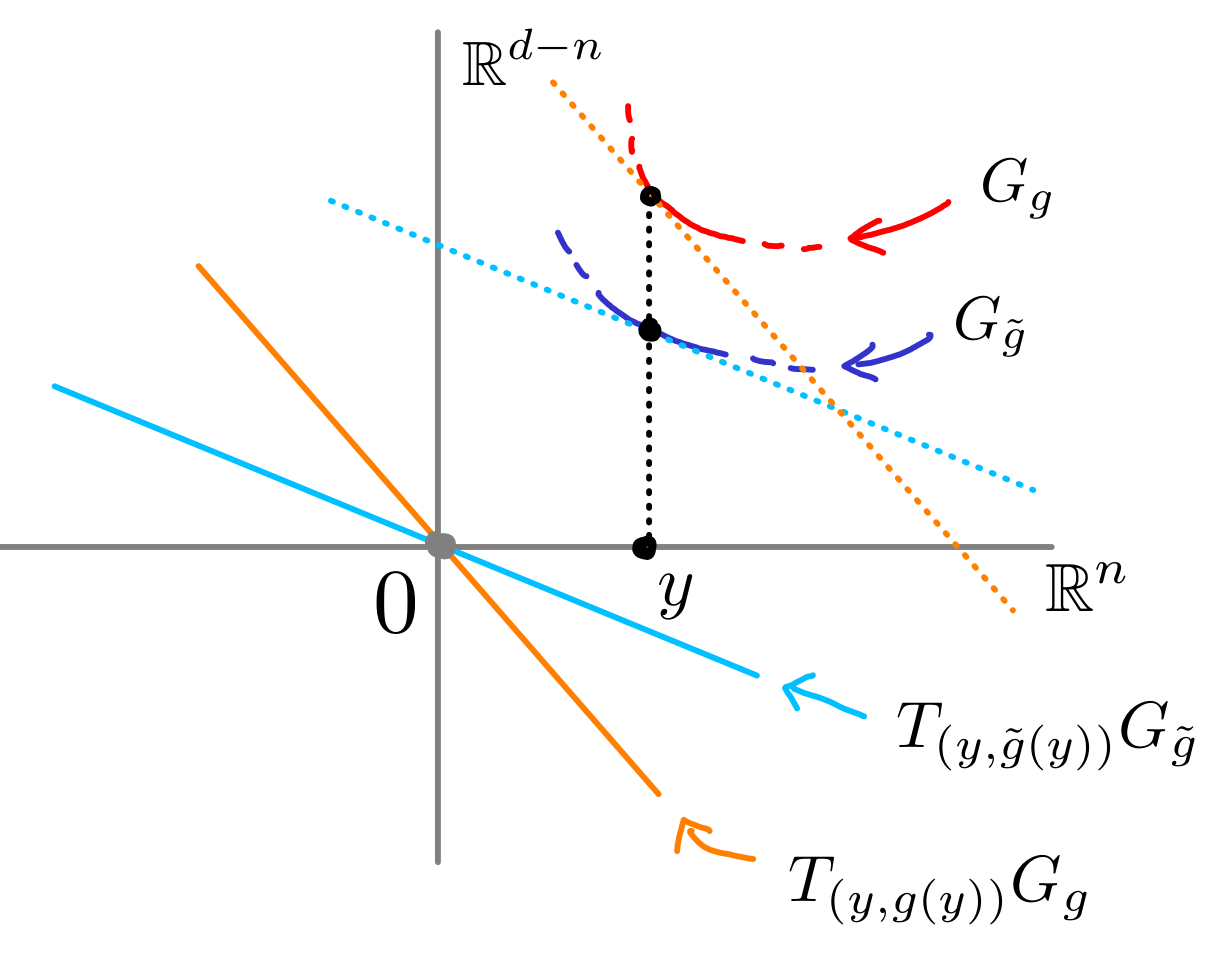}
		\caption{Illustration of the settings of Lemma~\ref{lemma:angle_between_tangent_spaces}.}
		\label{fig:lemma_26}
	\end{figure}

	\begin{proof}	
		Because the graph $G_g$ of $g$ is parametrized by $(y, g(y))$, its derivative is 
		\begin{align} 
			J_{g(y)}= 
			\begin{pmatrix}
				\textrm{Id} 
				\\ 
				Dg (y) 
			\end{pmatrix} ,
			\nonumber
		\end{align} 
		where $\textrm{Id}$ denotes the entries of the identity matrix and $Dg (y) $ the entries of the Jacobian at $y$. Similarly, the derivative of the graph $G_{\tilde{g}}$ of $\tilde{g}$ is 
		\begin{align} 
			J_{\tilde{g}(y)}= 
			\begin{pmatrix}
				\textrm{Id} 
				\\ 
				D\tilde{g} (y) 
			\end{pmatrix} ,
			\nonumber
		\end{align}
		
		\begin{figure}[h!]
			\centering
			\includegraphics[width=0.5\textwidth]{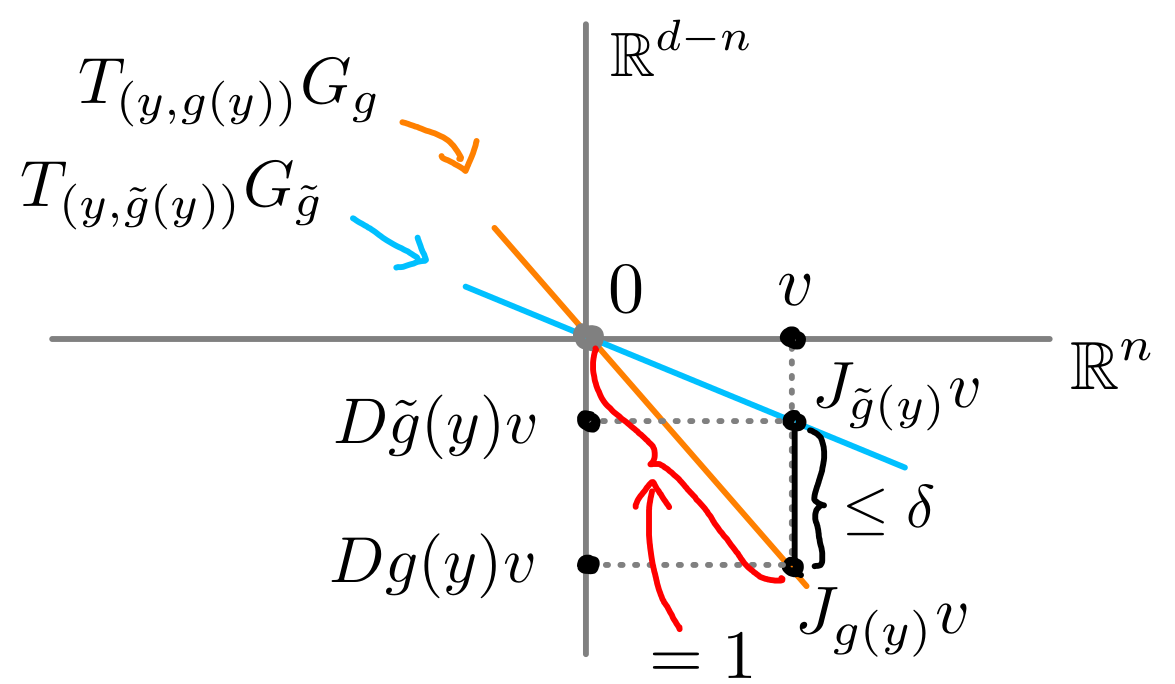}
			\caption{Notation in the proof of Lemma~\ref{lemma:angle_between_tangent_spaces}.}
			\label{fig:lemma_26_proof}
		\end{figure}
		
		For a vector $v\in\R^n\backslash\{0\}$, let $J_{g(y)}  v\in T_{(y, g(y))} G_g$ and $J_{\tilde{g}(y)}  v\in T_{(y, \tilde{g}(y))} G_{\tilde{g}}$  be its images under the linear maps $J_{g(y)}$ and $J_{\tilde{g}(y)}$. We illustrate the notation in Figure~\ref{fig:lemma_26_proof}.
		
		We rescale the vector $v$ in such a way that its image $J_{g(y)} v = \left(v,Dg(y) v\right)$ is a unit vector. This implies in particular that $|v|\leq 1$: $1 = |J_{g(y)} v|^2 = |v|^2 + |Dg(y)(v)|^2\geq |v|^2$.
		Then, since $ \| (Dg- D\tilde{g})(y)\|_2 \leq  \delta$,
		\begin{align*}
			| J_{g(y)} v -J_{\tilde{g} (y) }  v | = | ((Dg- D\tilde{g})(y)) (v) |\leq \delta |v| \leq \delta. 
		\end{align*}
		Finally, since $\delta<1$ and $|J_{g(y)} v|=1$,
		$\angle (J_{g(y)} v, J_{\tilde{g} (y) }  v) \leq\arcsin \delta.$
		(We illustrate this inequality in Figure~\ref{fig:lemma_26_proof2}.)
		
		\begin{figure}[h!]
			\centering
			\includegraphics[width=0.4\textwidth]{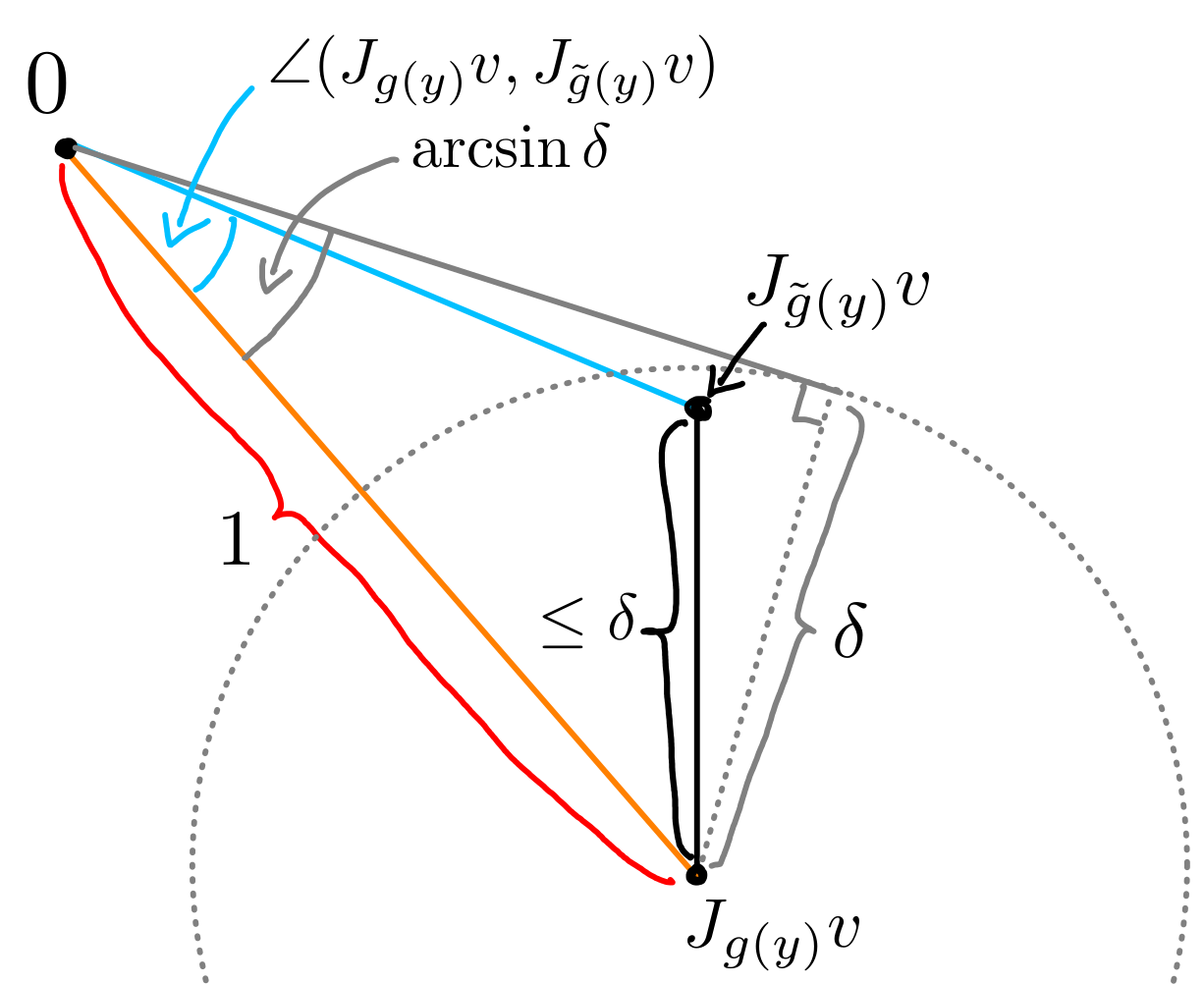}
			\caption{Since $\delta<1$ and $|J_{g(y)} v|=1$, the angle $\angle (J_{g(y)} v, J_{\tilde{g} (y) }  v)$ cannot be larger than $\arcsin \delta.$}
			\label{fig:lemma_26_proof2}
		\end{figure}
		Thus, $\min_{w \in T_{(y, \tilde{g}(y))} G_{\tilde{g}} \setminus \{0\}}
		\angle (J_{g(y)} v,w) \leq\arcsin \delta$. Since the choice of the vector $v$ was arbitrary,
		\[
		\angle (T_{(y, g(y))} G_{g},T_{(y, \tilde{g}(y))} G_{\tilde{g}}) = \max_{v\in \R^n\setminus \{0\}} \min_{w \in T_{(y, \tilde{g}(y))} G_{\tilde{g}} \setminus \{0\}}
		\angle (J_{g(y)} v,w) \leq\arcsin \delta.
		\]
	\end{proof}
	
	\begin{proof}[Proof of Lemma~\ref{boundangleTangentSpace}] 
		Recall that $F = f + \psi \cdot(\varphi_\varepsilon \ast f - f)$. 
		Due to Theorem \ref{Th23abHirsch},
		\begin{align*}
			DF-Df = D(F-f) = D(\psi \cdot(\varphi_\varepsilon \ast f - f)) = D\psi \cdot (\varphi_\varepsilon \ast f-f ) + \psi  \cdot(\varphi_\varepsilon \ast D f- D f ),
		\end{align*}
		and thus at any point $y\in C,$ we can use bounds \eqref{eq:boundDpsi}, \eqref{eq:AssumptionSmoothing2}, and \eqref{eq:AssumptionSmoothing3}, and the fact that $|\psi| \leq 1$, to bound the norm of $(DF-Df)(y)$:
		\begin{align*} 
			\| (DF- Df)(y)\|_2 &\leq \|  D\psi(y)  \|_2 \cdot |  (\varphi_\varepsilon \ast f-f )(y)|  + |\psi(y)| \cdot \|  (\varphi_\varepsilon \ast D f- D f )(y)  \|_2
			\\
			&\leq  L_{\psi,D \psi} \cdot \varepsilon  + 1 \cdot  \varepsilon.
		\end{align*} 
		Lemma~\ref{lemma:angle_between_tangent_spaces} then directly yields that
		\[
		\angle (\Tan_{(y,f (y))}G_f , \Tan_{(y,F(y))} G_F) = \angle (T_{(y,f (y))}G_f , T_{(y,F(y))} G_F) \leq\arcsin \left( L_{\psi,D \psi}  \varepsilon  +  \varepsilon\right) .
		\] 
	\end{proof} 
	
	\begin{proof}[Proof of Corollary \ref{CorHD}] 
		\begin{figure}[h!]
			\centering
			\includegraphics[width=.4\textwidth]{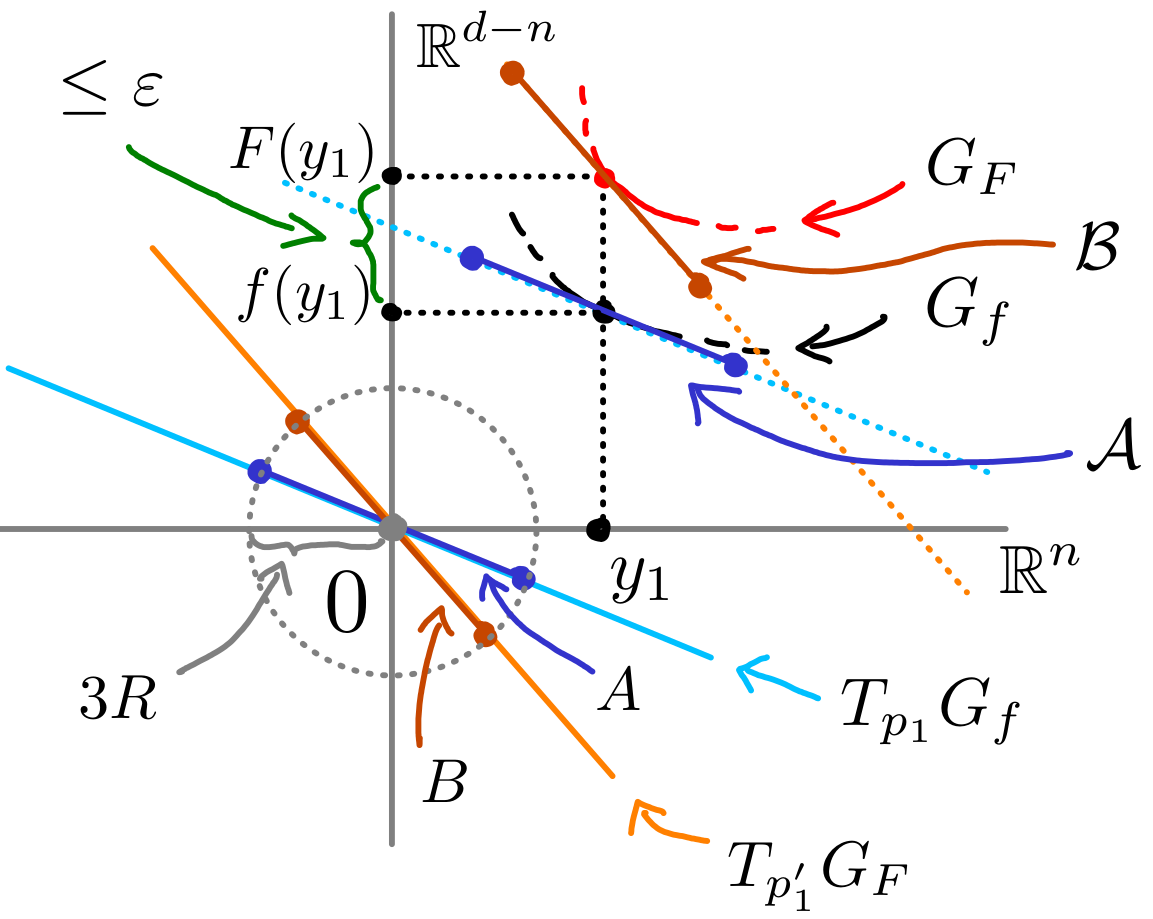}
			\caption{Illustration of the settings in the proof of Corollary \ref{CorHD}.}
			\label{fig:cor_27_proof}
		\end{figure}
		
		Let $A = \{T \in T_{p_1}G_f \mid |T| \leq 3R \}$ and $B = \{T \in T_{p_1'}G_F \mid |T| \leq 3R \}$, as illustrated in Figure~\ref{fig:cor_27_proof}.
		We first observe that
		\begin{align}\label{eq:cor_27_first_step}
			d_H(\mathcal{A}, \mathcal{B} )\leq |p_1'-p_1| + 	d_H(A,B).
		\end{align}	
		By assumption \eqref{eq:AssumptionSmoothing3}, and because $\abs{\psi}\leq 1$,
		\begin{align} 
			|p_1'-p_1| = | (y, F(y))- (y, f(y)) | = |F(y)-f(y) | = |\psi(y) \cdot (\varphi_{\varepsilon} \ast f-f)(y)| \leq 1 \cdot \varepsilon.
			\label{eq:trivialBoundPQ} 
		\end{align}
		To estimate the latter summand of equation~\eqref{eq:cor_27_first_step}, choose a vector $T \in T_{p_1}G_f \backslash \{0\}$, and let $T' \in  T_{p_1'}G_F \backslash \{0\}$ be a vector minimizing the angle between $T$ and $T_{p_1'}G_F$, that is, $\angle T, T' = \min_{U \in T_{p_1'}G_F \setminus \{0\}} \angle T,U$. Without loss of generality, we may assume that $T\in A$ and $T'\in B$. Then
		\begin{align*}
			d(T,B)\leq d(T,T') \leq 3R \:\angle (T, T') = 3R \min_{U \in T_{p_1'}G_F \setminus \{0\}}\angle (T,U),
		\end{align*}
		and thus, thanks to Lemma \ref{boundangleTangentSpace}, the one-sided Hausdorff distance $d(A,B)$ between $A$ and $B$ is bounded by
		\begin{align*}
			d(A,B) = \max_{T\in A\setminus \{0\}} d(T,B) \leq 3R \underbrace{\max_{T\in A\setminus \{0\}} \min_{U \in T_{p_1'}G_F \setminus \{0\}}\angle (T,U)}_{= \angle(T_{p_1}G_f,T_{p_1'}G_F)}\leq 3R\arcsin \left( L_{\psi,D \psi}  \varepsilon  +  \varepsilon\right).
		\end{align*}
		Using the same argument, we obtain the same bound for the one-sided Hausdorff distance between $B$ and $A$. Thus, $d_H(A,B)\leq 3R\arcsin \left( L_{\psi,D \psi}  \varepsilon  +  \varepsilon\right).$ 
		
		To obtain the final bound, we use the fact that $ L_{\psi,D \psi}  \varepsilon  +  \varepsilon\leq 1$, and thus
		\[3R\arcsin \left( L_{\psi,D \psi}  \varepsilon  +  \varepsilon\right)\leq 3R\cdot  2\left( L_{\psi,D \psi}  \varepsilon  +  \varepsilon\right).\]
	\end{proof} 

		We need the following technical lemma to prove Proposition~\ref{prop:pointsFarAway}:
	\begin{lemma}\label{lemma:technical_rewrite}
		For two constants $R,\xi\in\R$, $R\neq 0$, the expression $\frac{1}{2R} + \xi$ can be rewritten as
		\[\frac{1}{2R}+\xi = \frac{1}{2R-z},\]
		with $z = \frac{4R^2\xi}{1+2R\xi}$.
	\end{lemma}
	\begin{proof}
		\begin{align*}
			&0 = 4R^2\xi - z - 2Rz\xi \nonumber\\
			\iff 		&0 = \frac{1}{2R} \left( 4R^2\xi - z - 2Rz\xi \right)  \nonumber\\
			\iff & 1 =   2R\xi - \frac{z}{2R} - z\xi   + 1 = \left(\frac{1}{2R}+\xi\right)\left(2R-z\right)
			\nonumber 
			\\
			\iff & \frac{1}{2R}+\xi = \frac{1}{2R-z}.
		\end{align*}
	\end{proof}

	\begin{proof} [Proof of Proposition~{\ref{prop:pointsFarAway}}]
		
		
We can assume without loss of generality that $|q'-p'| < 2R$.
		Indeed, by definition of the distance, $ d(q' , \Tan_{p'}\M') \leq  |q'-p'|$, and if $|q'-p'| \geq 2R$ then $|q'-p'| \leq \frac{|q'-p'|^2}{2R}$, leading to the bound~\eqref{eq:distance_ala_Federer_with_epsilons}.
		

		Let $\mathcal{A}$ and $\mathcal{B}$ be the $3R$-neighbourhoods of the point $p_1$ and $p_1'$ in the affine tangent spaces $\Tan_{p_1}G_f$ and $\Tan_{p_1'}G_F$ from Corollary~\ref{CorHD}. Since $|q'-p'| \leq 3R$,
		\[
		d(q' , p+\Tan_{p_1} G_f) =d(q' , p+\mathcal{A}) \quad \text{and} \quad d(q' , p+\Tan_{p_1'} G_F) = d(q' , p+\mathcal{B}).
		\]
		Using the triangle inequality, we obtain:
		\begin{align*}
			d(q' , \Tan_{p'}\M') = d(q' , p+\mathcal{B}) &\leq d(q' , p+\mathcal{A})+ d_H(\mathcal{A},\mathcal{B}).
		\end{align*}
		We bound the above summands as follows:
		\begin{itemize}
			\item $d(q' , p+\mathcal{A})=d(q' , \Tan_{p+p_1} \M)  \leq  \frac{|q'-(p+p_1)|^2}{2R}$ due to Theorem~\ref{th418Federer} and the fact that $\M$ has reach~$R$;
			\item $d_H(\mathcal{A},\mathcal{B}) \leq \varepsilon\left( 6R L_{\psi,D \psi}+6R+1 \right)$ by Corollary \ref{CorHD} and since $\varepsilon\leq\tfrac{1}{L_{\psi,D \psi}+1}$.
		\end{itemize}
		Thus,
		\[d(q' , \Tan_{p'}\M')  \leq  \frac{|q'-(p+p_1)|^2}{2R}   + \varepsilon\left( 6R L_{\psi,D \psi}+6R+1 \right).\]
		Furthermore, since due to assumption \eqref{eq:AssumptionSmoothing3}, and $\abs{\psi}\leq 1$,
		\begin{align*} 
			|p_1'-p_1| = | (y, F(y))- (y, f(y)) | = |F(y)-f(y) | = |\psi(y) \cdot (\varphi_{\varepsilon} \ast f-f)(y)| \leq 1 \cdot \varepsilon,
		\end{align*}
		the factor $|q'-(p+p_1)|^2$ can be bound as
		\begin{align*}
			|q'-(p+p_1)|^2&\leq \left(|q'-(\underbrace{p+p_1'}_{=p'})| + |p_1-p_1'|\right)^2\leq \left(|q'-p'| + \varepsilon \right)^2=|q'-p'|^2 + 2\underbrace{|q'-p'|}_{\leq 3R}\varepsilon+\varepsilon^2\\
			&\leq |q'-p'|^2 + 6R\varepsilon +\varepsilon^2.
		\end{align*}
		Thus,
		\begin{align} 
			d(q' , \Tan_{p'}\M')  &\leq \frac{|q'-p'|^2}{2R}  + \frac{\varepsilon^2}{2R}   + \varepsilon\left(6R L_{\psi,D \psi}+6R+4\right),
			\nonumber
		\end{align} 
		which is precisely inequality~\eqref{eq:distance_ala_Federer_with_epsilons}.
		
		If moreover $\varepsilon \leq R$, the above simplifies to 
		\begin{align} 
			d(q' , \Tan_{p'}\M')  &\leq \frac{|q'-p'|^2}{2R}  + \varepsilon\left(6R L_{\psi,D \psi}+6R+4.5\right)\leq \frac{|q'-p'|^2}{2R}  + 6\varepsilon\left(R L_{\psi,D \psi}+R+1\right).
			\nonumber
		\end{align} 
		If in addition $|q'-p'|\geq \beta$, then 
		\begin{align*} 
			d(q' , \Tan_{p'}\M')   
			&\leq 
			\frac{|q'-p'|^2}{2R} + \frac{|q'-p'|^2}{\beta^2} \cdot 6\varepsilon\left(R L_{\psi,D \psi}+R+1\right).
		\end{align*} 
		We have now established a bound of the form 
		\begin{align*} 
			d(q' , \Tan_{p'}\M')  
			&\leq |q'-p'|^2 \left( 
			\frac{1}{2R} + \xi \right), 
		\end{align*}
		with 
		\[ \xi 
		=\frac{6\varepsilon}{\beta^2} \left(R L_{\psi,D \psi}+R+1\right) .\] 
		Using Lemma~\ref{lemma:technical_rewrite}, we finally transform this expression into
		\begin{align} 
			d(q' , \Tan_{p'}\M')  
			&\leq 
			\frac{|q'-p'| ^2}{2R'} , 
			\nonumber
		\end{align} 
		with 
		\begin{align}
			2R' = 2 R - \frac{4R^2\xi}{1+2R\xi} = 2R\frac{1}{1+2R\xi}= \frac{2R}{1+\frac{12\varepsilon R}{\beta^2}\left(R L_{\psi,D \psi}+R+1\right)}.
			\nonumber
		\end{align} 
	\end{proof} 


	\begin{proof}[Proof of Lemma~\ref{lemma:Lipschitz_conts_psi}]
		Choose two points $x,y\in B(0,\sqrt{\delta R}/2)\cap T_\p\M$. Then
		\begin{align*} 
			|\psi(x)- \psi(y) | &=  \left | \psi_0 \left( \frac{ 8 x}{\sqrt{\delta R}} \right)  -  \psi_0 \left( \frac{ 8 y}{\sqrt{\delta R}} \right) \right|\\ 
			& \leq L_{\psi_0} \left | \frac{ 8 x}{\sqrt{\delta R}} - \frac{ 8 y}{\sqrt{\delta R}} \right | \\
			& = L_{\psi_0} \frac{ 8 }{\sqrt{\delta R}} \left |  x-y \right |.
		\end{align*} 
		For a vector $u \in \mathbb{R}^n$, the difference between directional derivatives $\partial_u \psi(x)- \partial_u \psi(y)$ can be bound as 
		\begin{align*} 
			|\partial_u \psi(x)- \partial_u \psi(y) |&=  \left |  \partial_u  \psi_0 \left( \frac{ 8 x}{\sqrt{\delta R}} \right)  - \partial_u   \psi_0 \left( \frac{ 8 y}{\sqrt{\delta R}} \right) \right| 
			\\
			&= \frac{ 8 }{\sqrt{\delta R}} \left|  \partial_u  \psi_0 ( x')|_{x'= \frac{ 8 x}{\sqrt{\delta R}}}     - \partial_u   \psi_0  (y')|_{y'=  \frac{ 8 y}{\sqrt{\delta R} }}  \right| 
			\tag{by the chain rule}
			\\
			&\leq \frac{ 8 }{\sqrt{\delta R}} L_{D\psi_0} \left | \frac{ 8 x}{\sqrt{\delta R}} - \frac{ 8 y}{\sqrt{\delta R}} \right |  
			\\
			&= \frac{64}{\delta R} L_{D\psi_0} |x-y|
		\end{align*} 
	\end{proof}

\begin{proof}[Proof of Lemma~\ref{lemma:final_approximation_distance_between_q'_and_tangent_space_of_p'}]
	We distinguish between two cases: $p'\notin B(\p, \sqrt{\delta R}/4)$ and $p'\in B(\p, \sqrt{\delta R}/4)$.

	\subparagraph{\emph{The case $p'\notin B(\p, \sqrt{\delta R}/4)$:}}
	If $p'\notin B(\p, \sqrt{\delta R}/4)$, then $p'=p$ and $ \Tan_{p'}\M' = \Tan_{p}\M$.
	
	If also $q'\notin B(\p, \sqrt{\delta R}/4)$, then $q'=q$, and Theorem~\ref{th418Federer} yields
	\[d(q', \Tan_{p'}\M') = d(q, \Tan_{p}\M) \leq \frac{\abs{p'-q'}^2}{2R}.\]
	
	If, on the other hand, $q'\in B(\p, \sqrt{\delta R}/4)$, then we distinghish two cases. Either $p'\in B(\p, \sqrt{\delta R}/2)$, in which case we apply Lemma~\ref{lem:ErrorLipschitzDer}, combined with the bound~\eqref{eq:Lipschitz_const_final_psi} on the Lipschitz constant of $\psi$ and its derivative, and the bound $\tfrac{1}{R-\delta}$ on the Lipschitz constant of $f$, and obtain
	\begin{align}\label{eq:prebound1Epsilon2}
		d(q', \Tan_{p'}\M')\leq\tfrac{1}{2}\left(3\cdot\frac{64}{\delta R} L_{D\psi_0}\rho + \frac{1}{R-\delta}\right)\abs{p'-q'}^2.
	\end{align}
	Or $p'\in B(\p, \sqrt{\delta R}/2)$, in which case $\abs{p'-q'}\geq \sqrt{\delta R}/4$. We approximate
	\begin{align*}
		d(q', \Tan_{p'}\M')&\leq \abs{q'-q} + d(q, \Tan_{p}\M) \leq \abs{q'-q}\cdot\frac{\abs{p'-q'}^2}{\abs{p'-q'}^2} +\frac{\abs{p'-q'}^2}{2R}\\
		&\leq \abs{p'-q'}^2 \left(\frac{1}{2R}+\frac{\abs{q'-q}}{\tfrac{\delta R}{16}}\right) = \abs{p'-q'}^2 \left(\frac{1}{2R}+\frac{16\abs{q'-q}}{\delta R}\right).
	\end{align*}
	At the same time, due to the definition \eqref{eq:smoothed_mfld_function} of $F$, and inequality \eqref{eq:AssumptionSmoothing1},
	\begin{align*}
		\abs{q'-q} =\abs{F(y)-f(y)}\leq \psi(y) \cdot \abs{(\varphi_\rho  \ast f - f)(y) }\leq\rho.
	\end{align*}
	Thus, $d(q', \Tan_{p'}\M')\leq  \abs{p'-q'}^2 \left(\frac{1}{2R}+\frac{16\rho}{\delta R}\right)$.
	Finally, we apply Lemma~\ref{lemma:technical_rewrite} with $\xi = \frac{16\rho}{\delta R}$:
	\begin{align*}
		\frac{1}{2R}+\frac{16\rho}{\delta R} = \frac{1}{2R-\frac{4R^2\xi}{1+2R\xi}} = \frac{1}{2R-2R\frac{32\rho}{\delta+32\rho}} = \frac{1}{2R\cdot\frac{\delta}{\delta+32\rho}}.
	\end{align*}
	Hence,
	\begin{align}\label{eq:prebound3Epsilon2}
		d(q', \Tan_{p'}\M')\leq \tfrac{1}{2}\frac{\delta+32\rho}{\delta R} \abs{p'-q'}^2.
	\end{align}
	\subparagraph{\emph{The case $p'\in B(\p, \sqrt{\delta R}/4)$:}}
	At first assume that $q'\in B(\p, \sqrt{\delta R}/2)$. Then due to Lemma~\ref{lem:ErrorLipschitzDer}, combined with the bound~\eqref{eq:Lipschitz_const_final_psi} on the Lipschitz constant of $\psi$ and its derivative, and the bound $\tfrac{1}{R-\delta}$ on the Lipschitz constant of $f$,
	\begin{align*}
		d(q', \Tan_{p'}\M')\leq\tfrac{1}{2}\left(3\cdot\frac{64}{\delta R} L_{D\psi_0}\rho + \frac{1}{R-\delta}\right)\abs{p'-q'}^2.\tag{\ref{eq:prebound1Epsilon2}}
	\end{align*}
	
	Finally, assume that $q'\notin B(\p, \sqrt{\delta R}/2)$. Then $q'=q$ and, necessarily, $\abs{q'-p'}\geq \sqrt{\delta R}/4$.
	Then Proposition~\ref{prop:pointsFarAway}, with $\beta = \sqrt{\delta R}/4$, yields
	\begin{align}\label{eq:prebound2Epsilon2}
		d(q' , \Tan_{p'} \M')   \leq  \frac{|p'-q'|^2}{2R'}
	\end{align}
	with 
	\begin{align}  R'= 
		\frac{R}{1+\frac{3\cdot64\rho}{\delta}\left(\tfrac{64}{\delta}L_{D\psi_0}+R+1\right)}.
		\nonumber
	\end{align}
	
	\subparagraph{Merging the bounds together}
	By combining the bounds \eqref{eq:prebound1Epsilon2}, \eqref{eq:prebound3Epsilon2}, and \eqref{eq:prebound2Epsilon2}, we conclude that 
	\begin{align*}
		d(q' , \Tan_{p'} \M')   \leq  \frac{|p'-q'|^2}{2R}\cdot \Omega,
	\end{align*}
	with 
	\begin{align*}
		\Omega = \max\left\{ 1+32\tfrac{\rho}{\delta}, 3\cdot 64 L_{D\psi_0}\tfrac{\rho}{\delta} + \tfrac{1}{1-\tfrac{\delta}{R}}, 1+ 3\cdot 64\tfrac{\rho}{\delta}\left(\tfrac{64}{\delta}L_{D\psi_0} + R+ 1\right) \right\}.  
	\end{align*}
	Moreover, the first term in the expression $\Omega$ is smaller than the last: 
	\begin{align*}
		1+32\tfrac{\rho}{\delta}\leq 1+ 32\tfrac{\rho}{\delta} \cdot 6 + 64\tfrac{\rho}{\delta}\left(\tfrac{64}{\delta}L_{D\psi_0} + R\right),
	\end{align*} and thus
	\begin{align*}
		\Omega = \max\left\{3\cdot 64 L_{D\psi_0}\tfrac{\rho}{\delta} + \tfrac{1}{1-\tfrac{\delta}{R}}, 1+ 3\cdot 64\tfrac{\rho}{\delta}\left(\tfrac{64}{\delta}L_{D\psi_0} + R+ 1\right) \right\}.  
	\end{align*}
\end{proof}

\end{document}